\documentclass[conference]{IEEEtran}
\usepackage{hyperref}

\usepackage{times}
\usepackage{stmaryrd}
\usepackage{amssymb}
\usepackage{empheq}
\usepackage{verbatim}
\usepackage{amsmath,amsfonts,amsthm}
\usepackage{graphicx}
\usepackage{color}
\usepackage{xcolor,pgf,tikz,pgflibraryarrows,pgffor,pgflibrarysnakes}
\usepackage{subfigure}
\usepackage{microtype}
\usepackage{gastex,graphicx}
\usepackage{amsmath}
\usepackage{amsfonts}
\usepackage{amssymb}
\usepackage{tikz}

\usetikzlibrary{fit} % fitting shapes to coordinates
\usetikzlibrary{backgrounds} % drawing the background after the foreground

{\bfseries}{\rmfamily} 
{\bfseries}{\rmfamily}

\usepackage{blkarray}% http://ctan.org/pkg/blkarray
\usepackage{slashbox}
\usepgflibrary{shapes}
\usetikzlibrary{snakes,automata}
\usetikzlibrary{calc}

\tikzstyle{background}=[rectangle,fill=gray!10, inner sep=0.1cm, rounded corners=0mm]

\newcommand{\N}{\mathbb N}

\newcommand{\wB}{\Box^{\mathsf{w}}}
\newcommand{\wU}{\until^{\mathsf{w}}}
\newcommand{\wF}{\fut^{\mathsf{w}}}

\newcommand{\until}{\:\mathsf{U}}
\newcommand{\since}{\:\mathsf{S}}
\newcommand{\R}{\:\mathbb{R}}
\newcommand{\fut}{\Diamond}
\newcommand{\past}{\mbox{$\Diamond\hspace{-0.27cm}-$}}

\mathchardef\mhyphen="2D
\mathchardef\mhyph="2D

\newcommand{\nex}{\mathsf{O}}

\newcommand{\mtlfull}{\mathsf{MTL}[\until_I, \since_I]}

\newcommand{\mtlu}{\mathsf{MTL}[\until_I,\since_I]}
\newcommand{\mtl}{\mathsf{MTL}[\until_I]}
\newcommand{\mtlfutpw}{\mathsf{MTL}[\fut_I]}
\newcommand{\mtlfutp}{\mathsf{MTL}[\fut_I, \past_I]}
\newcommand{\mtlsns}{\mathsf{MTL}[\until_I,\since_{np}]}
\newcommand{\mtluns}{\mathsf{MTL}[\until_{np},\since_I]}

\newcommand{\mitlfp}{\mathsf{MTL}[\fut_{np},\past_{np}]}

\newcommand{\gmtlfull}{\mathsf{MTL}[\until_I, \since_I]}

\newcommand{\oomit}[1]{}

\newtheorem{theorem}{Theorem}

\newtheorem{lemma}{Lemma}

\usepackage{tikz}
\tikzstyle{nloc}=[draw, text badly centered, rectangle, rounded corners, minimum size=2em,inner sep=0.5em]
\tikzstyle{loc}=[draw,rectangle,minimum size=1.4em,inner sep=0em]
\tikzstyle{traNP}=[-latex, rounded corners]
\tikzstyle{traNP2}=[-latex, dashed, rounded corners]

\IEEEoverridecommandlockouts

\begin{document}
\title{Partially Punctual Metric Temporal Logic is Decidable}
\author{\IEEEauthorblockN{Khushraj Madnani}
 \IEEEauthorblockA{IIT Bombay}
  \and
  \IEEEauthorblockN{Shankara Narayanan Krishna}
  \IEEEauthorblockA{IIT Bombay}
  \and
  \IEEEauthorblockN{Paritosh K Pandya}
  \IEEEauthorblockA{TIFR}
}

\maketitle

Metric Temporal Logic $\mathsf{MTL}[\until_I,\since_I]$ is one of the most studied real time logics. It exhibits considerable diversity 
in expressiveness and decidability properties based on the permitted set of modalities and the nature of time interval constraints $I$. 
Henzinger et al., in their seminal paper showed that the non-punctual fragment of $\mathsf{MTL}$ called 
$\mathsf{MITL}$ is decidable. 
In this paper, we sharpen this decidability result by showing that the partially punctual fragment of $\mathsf{MTL}$ (denoted $\mathsf{PMTL}$) is decidable over strictly monotonic finite point wise time. In this fragment, we allow either punctual future modalities, or punctual past modalities, but never both together. We give two satisfiability preserving reductions from $\mathsf{PMTL}$ to the decidable logic $\mathsf{MTL}[\until_I]$. The first reduction uses simple projections, while the second reduction uses a novel technique of temporal projections with oversampling. We study the trade-off between the two reductions: while the second reduction allows the introduction of extra action points in the underlying model,  the equisatisfiable $\mathsf{MTL}[\until_I]$ formula obtained is exponentially succinct than the 
one obtained via the first reduction, where no oversampling of the underlying model is needed. 
We also show that $\mathsf{PMTL}$ is strictly more expressive than the fragments $\mathsf{MTL}[\until_I,\since]$ and $\mathsf{MTL}[\until,\since_I]$. 
\section{Introduction}
Metric Temporal Logic $\mathsf{MTL}$ is a well established logic useful for specifying quantitative properties of real time systems. The main modalities of  $\mathsf{MTL}$ are $\until_I$ (read ``until $I$'') and 
$\since_I$ (read ``since $I$''), where $I$ is a time interval with end points in $\mathbb{N}$. These 
formulae are interpreted over timed behaviours or timed words.  
A formula $a \until_{[2,3]} b$ holds at a position $i$ of a timed word iff there is a position $j$ 
strictly in the future of $i$ where $b$ holds, and 
at all intermediate positions between $i$ and $j$, $a$ holds good; moreover, the difference 
in the time stamps of $i$ and $j$ must lie in the interval [2,3].  
Similarly, $a \since_{[2,3]} b$ holds good at a point $i$ iff there is a position $j$ strictly 
in the past of $i$ where $b$ holds, and at all intermediate positions between $i$ and $j$ $a$ holds; 
further, the difference in the time stamps between $i$ and $j$ lie in the interval [2,3]. 
The intervals $I$ can be bounded 
of the form $\langle l, u \rangle$, or  
 unbounded of the form $\langle l, \infty)$, with $l, u \in \mathbb{N}$, and $\langle$ represents 
left closed or left open, while $\rangle$ represents right closed or right open intervals. The unary modalities 
$\fut_I$ (read ``fut $I$'') and $\past_I$ (read ``past $I$'') are special cases of until and since: 
$\fut_I a=true \until_I a$ and $\past_I a=true \since_I a$.

 The satisfiability question for various fragments of $\mathsf{MTL}$ has evoked lot of interest and work 
 over the past years. In their seminal paper, Alur and Henzinger showed that the satisfiability of $\mathsf{MTL}[\until_I, \since_I]$ is  undecidable, while the satisfiability of the ``non-punctual'' fragment $\mathsf{MITL}$ of  $\mathsf{MTL}[\until_I, \since_I]$ is decidable. As the name suggests, the non-punctual fragment disallows punctual intervals $I$: 
 these are intervals of the form $[t,t]$. The satisfiability of the future only fragment of $\mathsf{MTL}$, viz., 
 $\mathsf{MTL}[\until_I]$ was open for a long time, till Ouaknine and Worrell \cite{OuaknineW05} showed its decidability via  a reduction to 1-clock  alternating timed automata. 
 Even though the logic $\mathsf{MTL}[\until_I, \since]$ is more expressive than $\mtl$, 
 it was shown to be decidable \cite{deepak08} by an equisatisfiable reduction to $\mtl$. The 
 decidability of the unary fragment $\mathsf{MTL}[\fut_I, \past_I]$ has remained open for a long time, 
 it was recently shown undecidable \cite{our}.  The only fragment whose decidability is unknown 
 is thus, the ``partially punctual fragment'' of $\mathsf{MTL}$, where 
 we allow punctualities only in the future or in the past modalities, but never in both. 
 The main result of this paper is the decidability of the partially punctual fragment of $\mathsf{MTL}$ 
 for finite strictly monotonic timed words; our results can be adapted to work for 
 weakly monotonic finite words.
 %\input{hier.tex}
     
%The expressiveness hierarchy of various fragments of $\mathsf{MTL}$ 
%emphasizing on non-punctual past $\mathsf{MTL}[\until_I, \since_{np}]$ 
%is shown in figure \ref{fig:mtl}. An arrow $A \rightarrow B$ stands for logic $B$ being more expressive than logic $A$.
%A counter part of figure \ref{fig:mtl}, emphasizing on non-punctual future, 
% $\mathsf{MTL}[\until_{np}, \since_{I}]$ can be drawn.
%All the fragments included inside the polygon are decidable, while the ones outside are not. 
%It can be seen from figure \ref{fig:mtl} that  the partially punctual fragment 
%$\mathsf{MTL}[\until_I, \since_{np}]$ is at the top of the expressiveness 
%landscape. We show that $\mathsf{MTL}[\until_{np},\since_I]$, the 
%non-punctual future fragment of $\mathsf{MTL}$  cannot express some properties 
%which can be captured by the punctual unary future fragment $\mathsf{MTL}[\fut_I]$;
%similarly, the non-punctual past fragment of $\mathsf{MTL}$, 
%viz.,   $\mathsf{MTL}[\until_{I},\since_{np}]$ cannot express some properties 
%that are captured by the punctual unary full fragment  $\mathsf{MTL}[\fut_I,\past_I]$;
%finally we show that non-punctuality is not all that useless: there are properties 
%expressible in $\mathsf{MTL}[\fut_{np}, \past_{np}]$, which cannot be captured
%by the fragment $\mathsf{MTL}[\until_I,\since]$, where no time is allowed in past modalities.

\section{Metric Temporal Logic}
%\noindent{\it \underline{Timed Words}}:
  Let $\Sigma$ be a finite set of propositions. A finite timed word over $\Sigma$ is a tuple
$\rho = (\sigma,\tau)$ where $\sigma$ and $\tau$ are sequences $\sigma_1\sigma_2\ldots\sigma_n$ and  $\tau_1\tau_2\ldots\tau_n$ respectively, with $\sigma_i \in 2^{\Sigma}-\emptyset$,  and $t_i \in \R_{\geq 0}$
 for $1 \leq i \leq n$.   Let  $dom(\rho)$ be the set of positions $\{1,2,\ldots,n\}$ in the timed word.
Let $\Sigma=\{a,b\}$. An example of a timed word is $(\{a,b\},0.3)(\{b\}, 0.7)(\{a\},1.1)$.
$\rho$ is strictly monotonic iff $t_i < t_{i+1}$ for all $i,i+1 \in dom(\rho)$. 
Otherwise, it is weakly monotonic.
 Given $\Sigma$,  the formulae of $\mathsf{MTL}$ are built from $\Sigma$  using boolean connectives and 
time constrained versions of the modalities $\until$ and  $\since$ as follows: \\
$\varphi::=a (\in \Sigma)~|true~|\varphi \wedge \varphi~|~\neg \varphi~|
~\varphi \until_I \varphi~|~\varphi \since_I \varphi$\\
where  $I$ is an open, half-open or closed interval 
with end points in $\N \cup \{\infty\}$.    

Formulae of $\mathsf{MTL}$
are interpreted over timed words 
over a chosen set of propositions. 
Let $\varphi$ be an $\mathsf{MTL}$ formula.
If $\varphi$ is interpreted over timed words over  
$\Delta$, then we say that $\varphi$ is interpreted over $\Delta$.
Note that this is different from saying $\varphi$ is built from 
a set of propositions $\Sigma$: this just means that the propositions 
in $\varphi$ are taken from $\Sigma$.

%\begin{definition}
\label{point}
Given a finite  timed  word $\rho$, and an $\mathsf{MTL}$ formula $\varphi$, in the pointwise semantics, 
the temporal connectives of $\varphi$ quantify over a finite set of positions 
in $\rho$.  For an alphabet $\Sigma$, a timed word $\rho=(\sigma, \tau)$, a position 
$i \in dom(\rho)$, and an $\mathsf{MTL}$ formula $\varphi$, the satisfaction of $\varphi$ at a position $i$ 
of $\rho$ is denoted $(\rho, i) \models \varphi$, and is defined as follows:

\noindent $\rho, i \models a$  $\leftrightarrow$  $a \in \sigma_{i}$\\
$\rho,i  \models \neg \varphi$ $\leftrightarrow$  $\rho,i \nvDash  \varphi$\\
$\rho,i \models \varphi_{1} \wedge \varphi_{2}$   $\leftrightarrow$ 
$\rho,i \models \varphi_{1}$ 
and $\rho,i\ \models\ \varphi_{2}$\\
$\rho,i\ \models\ \varphi_{1} \until_{I} \varphi_{2}$  $\leftrightarrow$  $\exists j > i$, 
$\rho,j\ \models\ \varphi_{2}, t_{j} - t_{i} \in I$,\\
$~~~~~~~~~~~~~~~~~~~~~~~~~~~~~$ and  $\rho,k\ \models\ \varphi_{1}$ $\forall$ $i< k <j$\\
$\rho,i\ \models\ \varphi_{1} \since_{I} \varphi_{2}$  $\leftrightarrow$  $\exists\ j < i$,  
 $\rho,j\ \models\ \varphi_{2}$, $t_{i} - t_{j} \in I$,  \\
$~~~~~~~~~~~~~~~~~~~~~~~~~~~~~$ and  $\rho,k\ \models\ \varphi_{1}$ $\forall$  $j<k < i$
%\end{definition}

%Note that the above semantics applies irrespective of whether $\rho$ is a timed word over $\Sigma$ or $2^{\Sigma}$. 
\noindent $\rho$ satisfies $\varphi$ denoted $\rho \models \varphi$ 
iff $\rho,1 \models \varphi$. Let $L(\varphi)=\{\rho \mid \rho, 1 \models \varphi\}$. 
The set of all timed words over $\Sigma$ is denoted $T\Sigma^*$.

 A non-punctual interval  
has the form  $\langle a, b \rangle$ with $a \neq b$. 
We denote by $\mathsf{MTL}[\until_I,\since_{np}]$ the class of 
$\mathsf{MTL}$ formulae with non-punctual  
past modalities. Similarly, $\mathsf{MTL}[\until_{np},\since_I]$ is the class of 
$\mathsf{MTL}$ formulae with non-punctual
future modalities. The class of partially punctual 
$\mathsf{MTL}$ formulae, 
 $\mathsf{PMTL}$ consists of all formulae 
 with non-punctual future or non-punctual past.  
 $\mathsf{PMTL}=\mathsf{MITL} \cup \mathsf{MTL}[\until_{np},\since_I] \cup \mathsf{MTL}[\until_I, \since_{np}]$.

 Additional temporal connectives are defined in the standard way: 
we have the constrained future and past eventuality operators $\fut_I a \equiv true \until_I a$ and 
$\past_I a \equiv true \since_I a$, and their duals  
$\Box_I a \equiv \neg \fut_I \neg a$, 
$\boxminus_I a \equiv \neg \past_I \neg a$. Weak versions of  operators 
are defined as : $\wF a=a \vee \fut a, 
\wB a= a \wedge \Box a$, $a \wU b= b \vee [a \wedge (a \until b)]$.

%If $I$ is a non-singular interval of the form $\langle a, b \rangle$ with $a \neq b$ 
%then we denote the modalities by $\fut_{NS}$ ($\until_{NS}$), and $\past_{NS}$ ($\since_{NS}$). 
%\krish  I have removed BP,EP. 

 \section{Temporal Projections}
\label{remove-past}
In this section, we discuss the notion of ``temporal projections'' that are central to this paper. We discuss two 
kinds of temporal projections: simple projections, and oversampling projections. 
 
\subsection{Simple Extensions and Projections}
\label{simple}
 \noindent {\it \underline{$(\Sigma,X)$-simple extensions}}: 
 Let $\Sigma,X$ be finite sets of propositions such that $\Sigma \cap X=\emptyset$. 
 A $(\Sigma,X)$-simple extension is a timed word $\rho$ over $X \cup \Sigma$ such that at any point $i \in dom(\rho)$, $\sigma_i \cap \Sigma \ne \emptyset$. 
    For $\Sigma=\{a,b\}, X=\{c,d\}$,  
   $(\{a\},0.2)(\{b,c,d\},0.3)(\{b,d\},1.1)$ is a  $(\Sigma,X)$-simple extension. However,  
    $(\{a\},0.2)(\{c,d\},0.3)(\{b,d\},1.1)$ is not a $(\Sigma,X)$-simple extension for 
    the same choice of $\Sigma,X$, since for the position $i=2$, $\{c,d\} \cap \Sigma = \emptyset$. \\
\noindent{\it \underline{Simple Projections}}:
Consider a $(\Sigma,X)$-simple extension $\rho$. 
We define the {\it simple projection} of $\rho$ with respect to $X$, denoted $\rho \setminus X$ as the word obtained by %restricting $\rho$ to $\Sigma$.  $\rho \setminus X$ is obtained 
 erasing the symbols of $X$ from each $\sigma_i$.  Note that $dom(\rho)=dom(\rho \setminus X)$.
For example, if $\Sigma=\{a,c\}$, $X=\{b\}$, and 
$\rho=(\{a,b,c\},0.2)(\{b,c\},1)(\{c\},1.3)$, then 
$\rho \setminus X=(\{a,c\},0.2)(\{c\},1)(\{c\},1.3)$. 
$\rho \setminus X$ is thus, a timed word over $\Sigma$.
If the underlying word $\rho$ is {\it not} a $(\Sigma,X)$-simple extension, then 
the simple projection  of $\rho$ with respect to $X$ is {\it undefined}.

\noindent{\it{\underline{Equisatisfiability modulo Simple Projections}}}: 
Given $\mathsf{MTL}$ formulae $\psi$ and $\phi$, we say that $\phi$ is equisatisfiable to $\psi$ 
{\it modulo simple projections} iff there exist disjoint sets  $\Sigma,X$ such that 
\begin{enumerate}
\item $\phi$ is interpreted over $\Sigma$, and $\psi$ is interpreted over $\Sigma \cup X$,
    \item For any timed word $\rho$ over $\Sigma \cup X$, 
     $(\rho \models \psi) \rightarrow$ \\ $\rho$ is a  $(\Sigma,X)$-simple extension and $\rho \setminus X \models \phi,$
 \item For any timed word $\rho$ over $\Sigma$ such that  $\rho \models \phi$,  $\exists$  a $(\Sigma,X)$-simple extension $\rho'$ such that $\rho' \models \psi$, and \\
 $\rho' \setminus X = \rho$.
\end{enumerate}
We denote by $\phi=\exists X. \psi$, 
the fact that $\phi$ is equisatisfiable to $\psi$ modulo simple projections. \\
%It is clear that if  
%$\phi=\exists \Sigma. \psi$, then $\phi$ and $\psi$ are equisatisfiable.
\noindent{\it \underline{Extended Normal Form(ENF)}}: 
Given a formula $\varphi$ built from $\Sigma' \supseteq \Sigma$, 
the extended normal form of $\varphi$ with respect to $\Sigma$ denoted $ENF_{\Sigma}(\varphi)$  
 is the formula $\varphi \wedge \Box (\bigvee \Sigma)$. 
% \noindent{\it{\underline{Remark 1}}}. Let $\varphi$ be built from $\Sigma$. 
%   $\varphi$ is equivalent to $ENF_{\Sigma}(\varphi)$, when  interpreted on timed words over $\Sigma$. \\
% \noindent{\it{\underline{Remark 2}}}. Let $\psi$ be a formula.    
 %    When interpreted on  timed words  over  $\Sigma \cup X$, 
 % $L(ENF_{\Sigma}(\psi))=\{\rho' \models ENF_{\Sigma}(\psi) \mid \rho'$ is a  
 % $(\Sigma,X)$-simple extension$\}$.  

\begin{lemma}[Boolean Closure Lemma]
\label {lem:boolclosedequis}
Let $\varphi_1, \varphi_2$ be formulae built from  $\Sigma$. Let $\psi_1,\psi_2$ be formulae built 
from $\Sigma\cup X_1$ and $\Sigma \cup X_2$ respectively. 
Let $\Sigma_i=\Sigma \cup X_i$ for $i=1,2$, and let $X_1 \cap X_2 = \emptyset$. Then,
%interpreted over $\Sigma \cup 2^{X_1}$, $\Sigma \cup 2^{X_2}$, respectively such that $X_1 \cap X_2 = \emptyset$. Let $\Sigma_i = \Sigma \cup X_i$ for $i \in \{1,2\}$. Let $X=X_1 \cup X_2$. Then,\\
  $(\varphi_1 = \exists X_1. \psi_1$ and  $\varphi_2 = \exists  X_2. \psi_2)$ $\rightarrow$ $\varphi_1 \wedge \varphi_2 = \exists (X_1 \cup X_2). (\psi_1 \wedge \psi_2)$.
\end{lemma}
\begin{proof}
The proof can be found in Appendix \ref{proofs-1}.
\end{proof}

\subsection{ Flattening} 
\label{sec-flat}
Let $\varphi \in \mtlfull$ built from  
 $\Sigma$.
 %$ VAL = (2^\Sigma \cup act)$.
Given any sub-formula $\psi_i$ of $\varphi$, 
and a fresh symbol $b_i \notin \Sigma$, $T_i=\wB(\psi_i  \leftrightarrow b_i)$ 
is called a {\it temporal definition} and  $b_i$ is called a {\it witness}. 
Let $\psi=\varphi[b_i /\psi_i]$ be the formula 
obtained by replacing all occurrences of $\psi_i$ in $\varphi$, with the witness $b_i$.    
Flattening is done recursively until we have replaced all future/past modalities of interest
with witness variables, obtaining $\varphi_{flat}=\psi \wedge T$, where $T$ is the conjunction 
of all temporal definitions. Let $W$ be the set 
of all witness propositions.
For example, consider the formula $\varphi=a \until_{[0,3]}(c \since(\past_{[0,1]}d))$. 
Replacing the $\since, \past$ modalities with witness propositions 
$w_1$ and $w_2$ we get $\psi=a \until_{[0,3]}w_1$, 
along with the temporal definitions 
$T_1=\wB(w_1 \leftrightarrow (c \since w_2))$ and  
  $T_2=\wB(w_2 \leftrightarrow 
  \past_{[0,1]}d)$. Hence, $\varphi_{flat}=\psi \wedge T_1 \wedge T_2$ is obtained 
  by flattening the $\since,\past$ modalities from $\varphi$.  
   Here $W=\{w_1,w_2\}$.  Note that $\varphi_{flat}$ is a formula 
   built from $\Sigma \cup W$.
   
 Given  a timed word $\rho$ over $\Sigma$, 
 flattening marks precisely positions in $\rho$ satisfying $\psi_i$ with witnesses $b_i$. This marked word $\rho'$ 
 over $\Sigma \cup W$  satisfies $\varphi_{flat}$ iff  
$\rho \models \varphi$.   
    Hence, we have  
    $\varphi = \exists W. ENF_{\Sigma}(\varphi_{flat})$.
$ENF_{\Sigma}(\varphi_{flat})$ ensures that any timed word $\rho'$ over $\Sigma \cup W$ 
that satisfies $\varphi_{flat}$ is indeed a $(\Sigma,W)$-simple extension. 
   $L(ENF_{\Sigma}(\varphi_{flat}))$ is the set of all those $(\Sigma,W)$-simple extensions 
   $\rho'$ satisfying $\varphi_{flat}$ such that $\rho' \setminus W=L(\varphi)$.

      \subsection{Oversampled Behaviours and Projections}
      \label{oversample}

 \noindent {\it \underline{$(\Sigma,X)$-oversampled behaviours}}: 
 Let $\Sigma,X$ be finite sets of propositions such that $\Sigma \cap X=\emptyset$. 
 A $(\Sigma,X)$-oversampled behaviour is a timed word $\rho'=(\sigma', \tau')$ over $X \cup \Sigma$, such that 
 $\sigma'_1 \cap \Sigma \neq \emptyset$ and  $\sigma'_{|dom(\rho')|} \cap \Sigma \neq \emptyset$.
   For $\Sigma=\{a,b\}, X=\{c,d\}$,  
   $(\{a\},0.2)(\{c,d\},0.3)(\{a,b\},,0.7)(\{b,d\},1.1)$ 
 is a  $(\Sigma,X)$ oversampled behaviour, while 
   $(\{a\},0.2)(\{c,d\},0.3)(\{c\},1.1)$ is not.  
      If $\rho$ is  a $(\Sigma,X)$-oversampled behaviour, 
     then points $i$ where $\bigvee \Sigma$ is not true are called 
     {\it non-action points}.  Hence, in any  $(\Sigma,X)$-oversampled behaviour, 
     the first as well as the last points are action points.\\ 
% \noindent {\it \underline{Downsampling}}:
  %     Given an $(\Sigma,X)$-oversampled behaviour $\rho$,  we define an operation called {\it downsampling} with respect to $X$ 
 % on $\rho$ denoted $\rho\downarrow X$, which deletes from $\rho$, those $(\sigma_i,\tau_i)$ for which  $\sigma_i\subseteq X$.  
  %      For  $\Sigma=\{a,b\},X=\{c,d\}$, and 
   %    $\rho= (\{a\},0.2)(\{c,d\},0.3)(\emptyset,0.7)(\{b,d\},1.1)$, 
    %   $\rho \downarrow=(\{a\},0.2)(\{b,d\},1.1)$. \\
\noindent {\it \underline{Oversampled Projections}}:      
     Given a $(\Sigma,X)$-oversampled behaviour $\rho'=(\sigma',\tau')$, we define 
the {\it oversampled projection} of $\rho'$ with respect to $\Sigma$, 
 denoted $\rho' \downarrow X$ as the timed word obtained 
  by deleting points $i$ for which $\sigma'_i \cap \Sigma=\emptyset$, and then 
    erasing the symbols of $X$  from the remaining points $j$ ($\sigma'_j \cap \Sigma \neq \emptyset$).
  The result of oversampling,   $\rho$=$\rho'\downarrow X$ is a timed word over $\Sigma$.  
%Let $\rho'$ be a $(\Sigma,X)$-oversampled behaviour and let $\rho$ be a timed word 
%over $\Sigma$ such that $\rho' \downarrow X=\rho$.  
If  $\rho=\rho'\downarrow X$, there exists a strictly increasing function 
$f:\{1,2,\dots,n\} \rightarrow \{1,2,\dots,m\}$ such that $n=|dom(\rho)|$, $m=|dom(\rho')|$, and  
\begin{itemize}
\item $f(1)=1$, $\sigma_1=\sigma'_1 \cap \Sigma$, $\tau_1=\tau'_1$, and 
\item $f(n)=m$, $\sigma_n=\sigma'_m \cap \Sigma$, $\tau_n=\tau'_m$, and 
 \item For $1 \leq i \leq n-1$, $f(i)=j$ and $f(i+1)=k$ iff 
 \begin{itemize}
 \item $\sigma_i=\sigma'_j \cap \Sigma$,  
 %or $\sigma_i=\emptyset$, $\emptyset \neq \sigma'_j \subseteq Y$,
  and $\tau_i=\tau'_j$,
 \item $\sigma_{i+1}=\sigma'_k \cap \Sigma$, 
 %or $\sigma_{i+1}=\emptyset$, 
 %$\emptyset \neq \sigma'_k \subseteq Y$ 
 and   $\tau_{i+1}=\tau'_k$, 
 \item For all $j < l < k$, $\sigma'_l \cap \Sigma=\emptyset$. 
 \end{itemize}
 \end{itemize}
 
For ${\rho'=(\{a\},0.2)(\{a,c\},0.7)(\{c\},0.9)(\{b,d\},1.1)}$, a $(\Sigma,X)$-oversampled behaviour for ${\Sigma=\{a,b\}, X=\{c,d\}}$,  
we have  $\rho' \downarrow  X=(\{a\},0.2)(\{a\},0.7)(\{b\},1.1)$. We have ${f:\{1,2,3\} \rightarrow \{1,2,3,4\}}$ with 
${f(1)=1}, {f(2)=2}$, and ${f(3)=4}$.

\noindent{\it \underline{Equisatisfiability modulo Oversampled Projections}}:
Given $\mathsf{MTL}$ formulae $\psi$ and $\phi$, we say that $\phi$ is equisatisfiable to $\psi$ 
{\it modulo oversampled projections} iff there exist disjoint sets $X, \Sigma$ such that 
%$\phi$ is defined over $VAL = (2^{\Sigma  \setminus  X} \cup \{act\})$
\begin{enumerate}
\item $\phi$ is interpreted over $\Sigma$, and $\psi$ over $\Sigma \cup X$,
    \item For any $(\Sigma,X)$-oversampled behaviour $\rho'$,\\ 
     $\rho' \models \psi \rightarrow \rho' \downarrow X  \models \phi$
 \item For any timed word $\rho$ over $\Sigma$ such that  $\rho \models \phi$,  there exists  a $(\Sigma,X)$-oversampled behaviour  $\rho'$  such that $\rho' \models \psi$, and 
 $\rho'\downarrow X  = \rho$. 
\end{enumerate}
We denote by 
$\phi=\exists  \downarrow X. \psi$
 the fact that $\phi$ is equisatisfiable to $\psi$ modulo  oversampled projections. 
The above conditions establish the existence of {\it some} $(\Sigma,X)$-oversampled behaviour $\rho'$ 
corresponding to $\rho$ that satisfies $\psi$, when $\rho$ satisfies $\varphi$.
If condition 3 above holds for all possible $(\Sigma,X)$-oversampled behaviours, i.e, 
\begin{itemize}
\item if for any timed word $\rho$ over $\Sigma$ such that  $\rho \models \varphi$, all $(\Sigma,X)$-oversampled behaviours  $\rho'$  for which \\$\rho'\downarrow X = \rho$ satisfy $\psi$,
\end{itemize}
then we say that $\varphi$ and $\psi$ are {\it equivalent modulo oversampled projections} and denote it by 
$\varphi= \forall \downarrow. \psi$

\noindent{\it \underline{Oversampled Normal Form}} (ONF): 
Let $\psi$ be a formula built from $\Sigma \cup X$.  
Let $act$ denote $\bigvee \Sigma$. 
The {\it oversampled normal form} with respect to $\Sigma$ 
of $\psi$ denoted $ONF_\Sigma(\psi)$ is obtained by  replacing recursively 
\begin{itemize}
\item all subformulae of the form $a \in \Sigma$ by $a \wedge act$,
\item all subformulae 
of the form $\phi_i \until_I \phi_j$  with  \\ $(act \rightarrow ONF_{\Sigma}(\phi_i)) \until_I (ONF_{\Sigma}(\phi_j)\wedge act)$,
\item all subformulae of the form  $\phi_i \since_I \phi_j$ with \\ $(act \rightarrow ONF_{\Sigma}(\phi_i)) \since_I(ONF_{\Sigma}(\phi_j) \wedge act)$.
\item all subformulae of the form $\Box_I \phi$ with \\ $\Box_I(act \rightarrow ONF_{\Sigma}(\phi))$, and 
all subformulae of the form $\fut_I \phi$ with 
$\fut_I(\phi \wedge act)$.
      \end{itemize}  
and conjuncting the resultant formulae with 
$act \wedge (\Box \bot \rightarrow act)$.

 Let  $\psi=
 \varphi_1 \until_I (\varphi_2 \wedge \Box \varphi_3)$, 
and $\zeta_i{=}ONF_{\Sigma}(\varphi_i)$ for $i{=}1,2,3$. \\
Then 
  $ONF_{\Sigma}(\psi){=}
  (act {\rightarrow} \zeta_1) \until_I (act \wedge [\zeta_2 \wedge \Box (act {\rightarrow} \zeta_3)])
 \wedge \\  act \wedge (\Box \bot {\rightarrow} act)$
  where 
 $act$ denotes $\bigvee \Sigma$. 
Proofs of Lemmas \ref{lem:gen}, \ref{onf1}  and \ref{lem:boolclosedequis-2}
can be found in Appendices \ref{proofs-2}, \ref{proofs:onf1} and  
\ref{proofs-3}.
  \begin{lemma}[Oversampling Closure Lemma]
\label{lem:gen}
Let $\varphi$ be a formula built from $\Sigma$. 
Then $\varphi= \forall \downarrow. ONF_\Sigma(\varphi)$. 
\end{lemma}
%\begin{proof}
%The proof can be found in Appendix \ref{proofs-2}.
%\end{proof}

\begin{lemma}
\label{onf1}
Let $\varphi$ be a formula built from $\Sigma$ and let 
 $\zeta=ONF_{\Sigma}(\varphi)$. Then,  
 $\zeta=\forall \downarrow \zeta$.
\end{lemma}
 %\begin{proof}
% Follows from Lemma \ref{lem:gen} and  equivalence 
 %of $\zeta$ and $ONF_{\Sigma}(\zeta)$.
% \end{proof}

\begin{lemma}
\label {lem:boolclosedequis-2}
Consider  formulae $\varphi_1, \varphi_2$ built from $\Sigma$.   
Let $\psi_1, \psi_2$ be formulae built from $\Sigma \cup X_1$ and $\Sigma \cup X_2$ respectively.
Let $X=X_1 \cup X_2$, $\Sigma_i=\Sigma \cup X_i$ for $i=1,2$, and  $X_1 \cap X_2=\emptyset$.\\
Let $\zeta_1= ONF_{\Sigma_1}(\psi_1)$ and $\zeta_2= ONF_{\Sigma_2}(\psi_2)$. Then, \\ 
  $\varphi_1 = \exists \downarrow X_1. \zeta_1$ and  
  $\varphi_2 = \exists \downarrow X_2. \zeta_2 \rightarrow$\\
  $\varphi_1 \wedge \varphi_2 = \exists \downarrow X. 
  (\zeta_1 \wedge \zeta_2)$.
  \end{lemma}
%\begin{proof}
%The proof can be found in the Appendix \ref{proofs-3}.
%\end{proof}

 \begin{lemma}
\label{lem:flatgen}
Let $\varphi \in \mtlfull$ be built from $\Sigma$, and $W$ be the set of witness variables obtained 
while flattening $\varphi$. 
Then $\varphi=\exists \downarrow W. ONF_{\Sigma}(\varphi_{flat})$.
\end{lemma}

% \subsection{Figures for Lemma \ref{lemmapast}}
% \label{figs1}
\begin{figure*}[th]
 \begin{center}
 \begin{picture}(45,9)(20,-6)
 \thicklines
 \thicklines
    \drawline[AHnb=0,ATnb=0](-15,0)(105,0)
    \drawline[AHnb=0,ATnb=0](-15,-2)(-15,2)
    \drawline[AHnb=0,ATnb=0](105,-2)(105,2)
    \put(-5,2){\tiny{$a$}}
    \put (-5,-4){$\tau_{first_a}$}
    \put (-15,-4	){\tiny{0}}
     \drawline[AHnb=0,ATnb=0](-5,1)(-5,-1)
    \put(50,2){\tiny{$a$}}
    \put (50,-4){$\tau_{last_a}$}
    \put (82,-4){\tiny{$\tau_{last_a}+u$}}
      \drawline[AHnb=0,ATnb=0](50,1)(50,-1)
    %\drawline[AHnb=0,ATnb=0](-5,-9)(50,-9)
    \drawline[AHnb=0,ATnb=0](-15,4.5)(15,4.5)
    \drawline[AHnb=0,ATnb=0](-15,5.5)(-15,3.5)
    %\drawline[AHnb=0,ATnb=0](15,5.5)(15,3.5)
    \put(14.6,4){)}
    \drawline[AHnb=1,ATnb=0](85.5,4.5)(105,4.5)
    \drawline[AHnb=0,ATnb=0](85.5,5.5)(85.5,3.5)
    \drawline[AHnb=0,dash={0.25}0	](85.5,3.5)(85.5,-5)
   \drawline[AHnb=0,dash={0.25}0	](14.6,4)(14.6,-5)
      %\drawline[AHnb=0,ATnb=0](105,5.5)(105,3.5)
    %\drawline[AHnb=0,ATnb=0](-5,-5)(-5,-8)
    \put(13,-4){\tiny{$\tau_{first_a}+l$}}
 \put(92,6){$\neg \past_{[l,u)} a$}
    \put(0,6){$\neg \past_{[l,u)} a$}
   % \put(30,-4){\tiny{$\tau_j+u$}}
    %\put(30,-6.5){)}
  
   %\drawline[AHnb=0,dash={0.25}0	](30.8,0)(30.8,-6.5)
    %\put(70,-9){[}
  %  \put(85,-9){)}
     %\put(70,-11){\tiny{$\tau_k+l$}}
     %\put(85,-11){\tiny{$\tau_k+u$}}
     %\drawline[AHnb=0,ATnb=0	](70.5,0)(70.5,-9)
    % \drawline[AHnb=0,ATnb=0	](85.8,0)(85.8,-9) 
%     \drawline [AHnb=0,ATnb=0	](40.5,5) (80.5,5)
   % \drawline[AHnb=0,dash={0.25}0](70.5,-8.8)(85.5,-8.8)
    % \drawline[AHnb=0,ATnb=0](50,-5)(50,-8)
    %\drawline[AHnb=0,dash={0.25}0](15.5,-6)(30.5,-6)
    %\put(15,-6.5){)}
    %\node[Nframe=n,fillcolor=white](A)(-16,6){}
    %\node[Nframe=n,fillcolor=white](B)(16,6){}
    \end{picture}
 \caption{Cases (a) and (b) of Lemma \ref{lemmapast} : $\neg \past_{[l,u)}a$ holds in $[0, \tau_{first_a}+l)$ and $[\tau_{last_a}+u,\infty)$}
 \label{fig:case12}
 \end{center}
 \end{figure*}
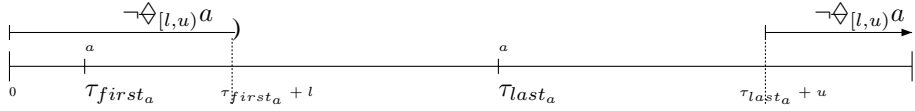
 \begin{figure*}[th]
   \begin{center}
   \begin{picture}(45,15)(20,-9)
   \thicklines
    \drawline[AHnb=0,ATnb=0](-15,0)(105,0)
    \drawline[AHnb=0,ATnb=0](-15,-2)(-15,2)
    \drawline[AHnb=0,ATnb=0](105,-2)(105,2)
    \put(-6.5,3){\tiny $a $}
    \put (-5,-4){$\tau_j$}
     \drawline[AHnb=0,ATnb=0](-5,-2)(-5,2)
    \put(13.5,3){\tiny $a $}
    \put (15,-4){$\tau_k$}
     \drawline[AHnb=0,ATnb=0](15,-2)(15,2)
    %\node[Nframe=n,fillcolor=white](A)(-16,6){}
    %\node[Nframe=n,fillcolor=white](B)(16,6){}
    \drawline[AHnb=0,ATnb=0](-5,-6.5)(15,-6.5)
    \drawline[AHnb=0,ATnb=0](-5,-5)(-5,-8)
     \drawline[AHnb=0,ATnb=0](15,-5)(15,-8)
    \drawline[AHnb=0,dash={0.25}0	](55.5,-6)(70.5,-6)
     \put(-3,-9){\tiny{$[u-l,u)$}}
    \put(55,-6.5){[}
   %\put(52.5,-8){\tiny{$\tau_{first_a}+l$}}
    \put(52.5,-8){\tiny{$\tau_j+l$}} %55 was the actual position
    \put(67.5,-8){\tiny{$\tau_j+u$}} %70 was the actual position
    \put(70,-6.5){)}
    \drawline[AHnb=0,dash={0.25}0	](55.5,0)(55.5,-6.5)
    \drawline[AHnb=0,ATnb=0	](70.5,7)(70.5,-6.5)
    \drawline[AHnb=0,ATnb=0](70.5,7)(75.5,7)
    \drawline[AHnb=0,dash={0.25}0](71.5,7)(71.5,0)
    \drawline[AHnb=0,dash={0.25}0](72.5,7)(72.5,0)
    \drawline[AHnb=0,dash={0.25}0](73.5,7)(73.5,0)
    \drawline[AHnb=0,dash={0.25}0](74.5,7)(74.5,0)
    \put(75,-9){[}
    \put(90,-9){)}
     \put(75,-11){\tiny{$\tau_k+l$}}
     \put(90,-11){\tiny{$\tau_k+u$}}
     \drawline[AHnb=0,ATnb=0	](75.5,7)(75.5,-9)
     \drawline[AHnb=0,ATnb=0	](90.8,0)(90.8,-9) 
 %    \drawline [AHnb=0,ATnb=0	](70.5,4) (75.5,4)
    \drawline[AHnb=0,dash={0.25}0	](75.5,-8.8)(90.5,-8.8)
%     \put(71.5,13){$\neg b$}
% \put (65,6.75){$y_0$}
% \put(60,4.5){[}
% \put(74.5,4.5){)}
% \put (76,10.75){$x_0$}
% \put(70.,8.5){[}
% \put(84.5,8.5){)}

%          \drawline [AHnb=0,ATnb=0	](60.5,5) (75.5,5)
%          \drawline [AHnb=0,ATnb=0	](70.5,9) (85.5,9)
  
   \end{picture}
   \caption{Case (c) Lemma \ref{lemmapast}:   $\neg \past_{[l,u)} a$ holds in shaded region}
   \label{fig:case3}
   \end{center}
   \end{figure*}
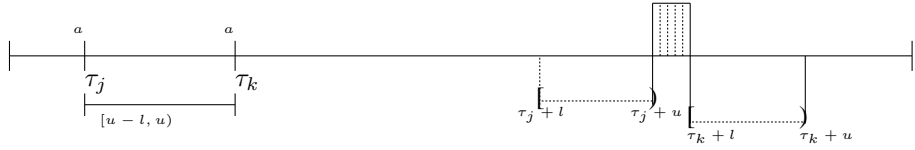

\section{ Decidability of $\mtlsns$}
In this section, we show that the class $\mtlsns$ is decidable, by giving a satisfiability preserving 
reduction  to $\mtl$. Given a timed word $\rho$, and a non-singular past modality 
of the form $\psi=\past_{\langle l, u \rangle}\varphi$, Lemma \ref{lemmapast} 
establishes a relationship between time stamps of the points in $\rho$  where $\psi$ holds and 
the time stamps of points where $\varphi$ holds in $\rho$ with respect to $l,u$.

\begin{lemma}
\label{lemmapast}
 Given a timed word $\rho=(\sigma, \tau)$ and  a point $i \in dom(\rho)$.
 Let $first_{\alpha}$ and $last_{\alpha}$ denote respectively the first and last occurrences of $\alpha \in \Sigma$ in $\rho$. 
  $\rho, i \models \neg (\past_{ \langle l, u \rangle }\alpha)$ iff 
 \begin{itemize}
  \item[(a)] $\tau_i \sim_1 \tau_{first_{\alpha}} + l$, where
  $\sim_1$ is $<$ when $\langle$ is $[$, and $\sim_1$ is $\leq$ when $\langle$ is $($, or
     \item[(b)] $\tau_i \sim_2 \tau_{last_{\alpha}} + u$, where 
     $\sim_2$ is $>$ when $\rangle$ is $]$, and $\sim_2$ is $\geq$ when $\rangle$ is $)$,or 
    %(when $\rangle$ is $\textbf{)}$,  $\tau_i \ge\ \tau_{last_{\textbf{a}}} + u)$    
  \item[(c)] $\tau_i \in \langle \tau_j + u,\tau_k + l \rangle$ for all points $j,k  (j<k )$ where $\alpha$ holds consecutively (that is there does not exist any point $z$, $j<z<k$ where $\alpha$ holds). 
  Note that in this case $\tau_j + u \le \tau_k + l$.
 \end{itemize}
\end{lemma}
\begin{proof}
We prove the lemma for intervals of the form $[l,u)$.  The proof can be extended for other type of intervals also.
Assume that $\rho, i \models \past_{[ l, u )}\alpha$. We then show that $\neg(\tau_i < \tau_{first_{\alpha}}+l)$ and 
$\neg(\tau_i \geq  \tau_{last_{\alpha}} +u)$ and $\neg(\tau_i \in [\tau_j+u, \tau_k+l))$ for consecutive points $j,k$ where $\alpha$ holds.
%Part 1 ($\leftarrow$) 
\begin{enumerate}
\item   Let $\tau_i < \tau_{first_{\alpha}}+l$.  
$\rho, i \models \past_{[ l, u )}\alpha$ implies that there is a point $i'$ such that $\tau_{i'} \in (\tau_i-u, \tau_i-l]$, such that $\rho, i' \models \alpha$. Then, $\tau_{i'} \leq \tau_i-l <\tau_{first_{\alpha}}$, contradicting 
that $first_{\alpha}$ is the first point where $\alpha$ holds.   
\item Let $\tau_i \geq \tau_{last_{\alpha}} +u$. Again, $\rho, i \models \past_{[ l, u )}\alpha$ implies that there is a point $i'$ such that 
$\tau_{i'} \in (\tau_i-u, \tau_i-l]$ such that $\rho,i' \models \alpha$. We then have $\tau_{i'} > \tau_i -u \geq \tau_{last_{\alpha}}$, contradicting 
that $last_{\alpha}$ is the last point where $\alpha$ holds. 
\item Assume that there exist consecutive points $j < k$ where $\alpha$ holds. Also, let $\tau_{i} \in [\tau_j+u, \tau_k+l)$.  
$\rho, i \models \past_{[ l, u )}\alpha$ implies that there exists a point $i'$ such that $\tau_{i'} \in (\tau_i-u, \tau_i-l]$ and $\rho,i' \models \alpha$.
Also, $\tau_i-u \in [\tau_j, \tau_k+(l-u))$  and $\tau_i-l \in [\tau_j+(u-l),\tau_k)$. This gives $\tau_j < \tau_{i'} < \tau_k$ contradicting the assumption that 
$j,k$ are consecutive points where $\alpha$ holds. 
\end{enumerate}
The converse can be found in Appendix \ref{proofs-6}.
Figure \ref{fig:case12} illustrates regions for 
cases (a) and (b), while Figure \ref{fig:case3} illustrates the region for case (c). In the rest of the paper, we refer 
to regions in case(a) as Region I, regions in case(b) as Region II and regions in case (c) as Region III.
%In Appendix \ref{proofs-6}.    
\end{proof}
In the rest of this section, we show the decidability of $\mathsf{MTL}[\until_I, \past_{np}]$ 
by reducing any formula   $\varphi \in \mathsf{MTL}[\until_I, \past_{np}]$ to a formula 
$\psi \in \mtl$. We have two techniques for this proof: one using oversampling projections, and the other, using simple projections. 
%\begin{itemize}
%\item The first approach is using oversampling projections.  
%Given $\varphi \in \mathsf{MTL}[\until_I, \past_{np}]$ built from $\Sigma$, we synthesize 
%a formula $\psi \in \mtl$, built from $\Sigma \cup X$ such that $\varphi=\exists \downarrow X.\psi$.
%For this, we show that given any timed word $\rho$ over $\Sigma$, 
%we can construct a $(\Sigma,X)$-oversampled behaviour $\rho'$ such that 
%$\rho' \models \psi$ and $\rho' \downarrow X=\rho$. Conversely, we also show that 
%for any $(\Sigma,X)$-oversampled behaviour $\rho'$ satisfying 
%$\psi$, $\rho' \downarrow X \models \varphi$. 
%The size of the newly constructed formula $\psi$ has only a  linear increase 
%over the size of $\varphi$. Lemmas \ref{remove-pastinf} and 
  %    Lemma \ref{past-b1} describe how $\psi$ is synthesized.

%\item The second approach is using simple projections.  
%Given $\varphi \in \mathsf{MTL}[\until_I, \past_{np}]$ built from $\Sigma$, we synthesize 
%a formula $\psi \in \mtl$, built from $\Sigma \cup X$ such that $\varphi=\exists X.\psi$.
%For this, we show that given any timed word $\rho$ over $\Sigma$, 
%we can construct a $(\Sigma,X)$-simple extension $\rho'$ such that 
%$\rho' \models \psi$ and $\rho' \setminus X=\rho$. Conversely, we also show that 
%for any $(\Sigma,X)$-simple extension $\rho'$ satisfying 
%$\psi$, $\rho' \setminus X \models \varphi$. 
%The size of the newly constructed formula $\psi$ has  an exponential increase 
%over the size of $\varphi$. Lemmas 
%\ref{remove-pastinf2}
 %and   Lemma \ref{past-b2} in section \ref{sp-main} describe how $\psi$ is synthesized.

%\end{itemize}
\subsection{Elimination of Past with Oversampled Projections}
\label{os-main}
%In section \ref{sp-main}, we gave a satisfiability preserving reduction eliminating past operators, using simple projections. The resultant $\mtl$ formula had size exponential in the input formula size. In this section, 
%using oversampled projections, we elucidate  a satisfiability preserving reduction  eliminating past,  where the size of the resultant $\mtl$ formula is linear in  the input formula size. 
In this section, given a formula $\varphi$ in $\mathsf{MTL}[\until_I, \past_{np}]$ built from $\Sigma$, we 
synthesize a formula  $\psi \in \mtl$ built from $\Sigma \cup X$ equisatisfiable to $\varphi$ 
{\it modulo oversampled projections},  whose size 
is \emph{linear} in $|\varphi|$. Starting with a timed word   $\rho$ over $\Sigma$, 
we synthesize an $(\Sigma,X)$-oversampled behaviour $\rho'$ 
such that $\rho \models \varphi$ iff $\rho' \models \psi$. 
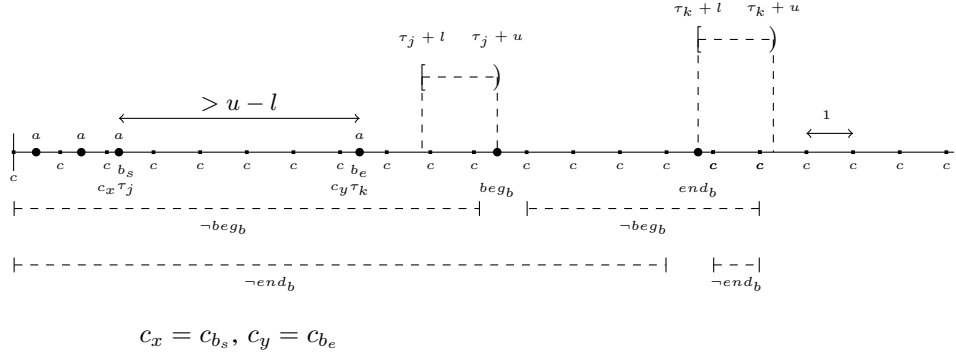
\begin{figure*}[t]
\begin{center}
\tikzstyle{RectObject}=[rectangle,fill=white,draw,line width=0.5mm]
\tikzstyle{line}=[draw]
\tikzstyle{arrow}=[draw, -latex] 
\begin{tikzpicture}
% a straight line segment
\foreach \x in {10.5}
{
\draw (\x,0) -- (\x+12.5,0);
\node at (\x,-0.34){\tiny $c$};

\node at (\x+0.62*2,-0.5){\tiny $c_x$};
\node at (\x+0.62*7,-0.5){\tiny $c_y$};
\node at (\x+3,-2.5){$c_x = c_{b_s}$, $c_y=c_{b_e}$};
\foreach \y in {0}{
\draw (\x+\y,-0.25)--(\x+\y,0.25);
}
\node[fill = black,draw = black,rectangle,inner sep=0.5 pt,label=below:{}] at (\x,0){};

\foreach \y in {0.62,1.24}
\node[fill = black,draw = black,rectangle,inner sep=0.5 pt,label=below:{\tiny{$c$}}] at (\x+\y,0){};

\foreach \y in {1.86,2.48,...,4.6}
\node[fill = black,draw = black,rectangle,inner sep=0.5 pt,label=below:{\tiny{$c$}}] at (\x+\y,0){};

%  \foreach \y in {6.82,7.44,...,9.5}
%  \node[fill = black,draw = black,rectangle,inner sep=0.5 pt,label=above:{\tiny{$c'$}}] at (\x+\y,0){};

%\node[fill = black,draw = black,rectangle,inner sep=0.5 pt,label=above:{\tiny{$c$}}] at (\x+7.44,0){};

\foreach \y in {4.96,5.54,...,6.5}
\node[fill = black,draw = black,rectangle,inner sep=0.5 pt,label=below:{\tiny{$c$}}] at (\x+\y,0){};

\foreach \y in {6.82,7.44,...,10.1}
{
\node[fill = black,draw = black,rectangle,inner sep=0.5 pt,label=below:{\tiny{$c$}}] at (\x+\y,0){};
% \node[fill = black,draw = black,rectangle,inner sep=0.5 pt,label=above:{\tiny{$g_b$}}] at (\x+\y,0){};
}

\foreach \y in {9.3,9.92,...,12.5}
{
\node[fill = black,draw = black,rectangle,inner sep=0.5 pt,label=below:{\tiny{$c$}}] at (\x+\y,0){};
}
\foreach \y in {0.3,0.9,1.4,4.6}	
\node[fill = black,draw = black,circle,inner sep=1 pt,label=above:{\tiny{$a$}}] at (\x+\y,0){};

\draw[<->](\x+10.54,0.25)--(\x+11.16,0.25);
\node at (\x+10.81,0.5){\tiny {$1$}};

% \foreach \y in {5.4,6.02,6.68}
% {
%  \node[fill = black,draw = black,rectangle,inner sep=0.5 pt,label=below:{\tiny{$z$}}] at (\x+\y,0){};
%  
%   \node at (\x+\y+0.31,0.35){\tiny $1$};
%  }
 
%  \draw[|-|] (\x+5.4,0.25)--(\x+6.02,0.25);
% \draw[|-|] (\x+6.02,0.25)--(\x+6.68,0.25);
% \draw[|-|] (\x+6.68,0.25)--(\x+7.3,0.25);
% 
% 

\node at (\x+1.5,-0.25){\tiny $b_s$};
\node at (\x+4.6,-0.25){\tiny $b_e$};

 \draw[<->](\x+4.6,0.45)--(\x+1.4,0.45);%changed

\node at (\x+3,0.65) {\small {$> u-l$}};
\node at (\x+1.5,-0.5){\tiny $\tau_j$};
\node at (\x+4.6,-0.5){\tiny $\tau_k$};

%\draw[<->] (\x+4.6,1.5)--(\x+4.6+4.5,1.5);
%\node at (\x+4.6+0.935,1.75){$l$};
\node at (\x+4.6+4.5,1.5) {$[$};	
\node at (\x+4.6+5.5,1.5) {$)$};
\draw[dashed] (\x+4.6+4.5,1.5)--(\x+4.6+5.5,1.5);
\draw[dashed] (\x+4.6+4.5,1.5)--(\x+4.6+4.5,0);
\draw[dashed] (\x+4.6+5.5,1.5)--(\x+4.6+5.5,0);
%\node at (\x+4.6+3,1.75){$u-l$};

%\draw[ultra thick, |-|] (\x+4.6+4.5,0.5)--(\x+0.93+5.5,0.5);
\node[fill = black,draw = black,circle,inner sep=1pt,label=below:{}] at (\x+0.93+5.5,0){};
\node at (\x+0.93+5.5,-0.5){\tiny $beg_b$};

\node at (\x+4.6+4.5,-0.5){\tiny $end_b$};

\draw[dashed,|-|] (\x,-0.75)--(\x+6.2,-0.75);

\draw[dashed,|-|] (\x+0.62*16,-0.75)--(\x+0.62*11,-0.75);

\node at (\x+9*0.31,-1){\tiny $\neg beg_b$};
\node at (\x+27*0.31,-1){\tiny $\neg beg_b$};

\draw[dashed,|-|] (\x,-1.5)--(\x+14*0.62,-1.5);

\draw[dashed,|-|] (\x+0.62*15,-1.5)--(\x+0.62*16,-1.5);

\node at (\x+11*0.31,-1.75){\tiny $\neg end_b$};
\node at (\x+31*0.31,-1.75){\tiny $\neg end_b$};

\node at (\x+0.93+5.5,1.5){\tiny {$\tau_j+u$}};
\node at (\x+0.93+4.5,1.5){\tiny {$\tau_j+l$}};
%\node at (\x+0.93+4.5,0.935){\tiny {$\end_b$}};

\node at (\x+4.6+4.5,1.9){\tiny {$\tau_k+l$}};
\node at (\x+4.6+5.5,1.9){\tiny {$\tau_k+u$}};
%\node at (\x+4.6+5.5,1.9){\tiny {$\beg_b$}};

\node[fill = black,draw = black,circle,inner sep=1pt,label=below:{}] at (\x+4.6+4.5,0){};
%\node at (\x+8.5,0.75) {$\neg b$};

%\draw[<->] (\x+0.93,1)--(\x+0.93+4.5,1);
%\node at (\x+0.93+0.935,0.935){$l$};
\node at (\x+0.93+4.5,1) {$[$};
\node at (\x+0.93+5.5,1) {$)$};
\draw[dashed] (\x+0.93+4.5,1)--(\x+0.93+5.5,1);
\draw[dashed] (\x+0.93+4.5,1)--(\x+0.93+4.5,0);
\draw[dashed] (\x+0.93+5.5,1)--(\x+0.93+5.5,0);
%\node at (\x+0.93+3,0.935){$u-l$};

}		
\end{tikzpicture}
\caption{Marking $[\tau_j+u,\tau_k+l)$ with $\neg b$}
\label{fig:os}
\end{center}
\end{figure*}

 \begin{enumerate}
 \item Start with a formula $\varphi \in \mathsf{MTL}[\until_I, \past_{np}]$ built from $\Sigma$, and a timed word $\rho$  
 over $\Sigma$, 
 \item Flatten $\varphi$ obtaining $\varphi_{flat}$. Let $W$ be the witness propositions used. 
 $\varphi_{flat}$ is a formula built from $\Sigma \cup W$, with $\Sigma \cap W=\emptyset$.
 \item  Let $T=\bigwedge_{i=1}^k T_i$ be the conjunction of all temporal definitions in $\varphi_{flat}$. 
  Each $T_i$ has the form $\wB(b \leftrightarrow \past_{\langle l, u \rangle}a)$, 
 with $l, u \in \R_{\geq 0} \cup \{\infty\}$, and $\bigwedge_{i=1}^k T_i$ is built from
 $\Sigma \cup W$.  
  $\varphi_{flat}=\psi \wedge T$, with $\psi \in \mtl$.
 We know from Lemma \ref{lem:flatgen} that 
   $\varphi=\exists \downarrow W. ONF_{\Sigma}(\varphi_{flat})$. 
  \item For $i=1,2, \dots, k$, let $\Sigma_i=  \Sigma\cup W \cup X_i$, where $X_i$ are 
  a set of fresh propositions, such that $X_i \cap X_j=\emptyset$ for $i \neq j$.   
  Synthesize a formula $\zeta_i=ONF_{\Sigma_i}(\varphi'_i) \in \mtl$ over $\Sigma_i$ 
    such that $ONF_{\Sigma}(T_i) =\exists \downarrow X_i. \zeta_i$.  
    \item Using Lemma \ref{lem:boolclosedequis-2},  
     $\bigwedge_{i=1}^k \zeta_i \in \mtl$ is such that\\  
   $ONF_{\Sigma}( \bigwedge_{i=1}^kT_i)=\exists \downarrow X. \bigwedge_{i=1}^k \zeta_i$, for $X=\bigcup_{i=1}^k X_i$.
        \end{enumerate}
 
 Lemma \ref{remove-pastinf} and 
      Lemma \ref{past-b1} show how to synthesize an equisatisfiable  formula in $\mtl$ 
      corresponding to $ONF_{\Sigma}(T_i)$.
       Lemma \ref{remove-pastinf} shows step 4 for intervals of the form $[l, \infty)$, while 
 Lemma \ref{past-b1} shows step 4 for bounded intervals of the form $[l,u)$. The results of these lemmas 
 can be extended to work for any interval $\langle l, u \rangle$. 
If all the past modalities involved have unbounded intervals, 
then we get an equivalent formula, as shown by Lemma \ref{remove-pastinf}.    
      
\begin{lemma}
\label{remove-pastinf}
Consider a temporal definition 
 $T=\wB[b \leftrightarrow \past_{[l, \infty)}a]$ built from $\Sigma \cup W$. 
 Then we can synthesize a formula $\psi \in \mtl$ built from  $\Sigma \cup W$ 
 equivalent to $ONF_{\Sigma}(T)$.
\end{lemma}
\begin{proof}
It can be shown that  $[\wB \alpha \vee \{\alpha \wU[(a \wedge act) \wedge \wB_{[0,l)}(act \rightarrow \neg b)] \}]
 \wB[(a \wedge act) \rightarrow \Box_{[l,\infty)}(act \rightarrow b)]\footnote{when $l=0$,  $\alpha \wU[a \wedge act \wedge \neg b]$}$ is equivalent to $ONF_{\Sigma}(T)$, 
 for $\alpha=(act \rightarrow (\neg a \wedge \neg b))$.
Details in Appendix \ref{proofs-4}.
  \end{proof}

\begin{lemma}
 \label{past-b1}
    Consider a temporal definition 
 $T=\wB[b \leftrightarrow \past_{[l, u)}a]$,  built from $\Sigma \cup W$. 
 Then we can synthesize a formula $\psi \in \mtl$ 
 built from  $\Sigma \cup W \cup X$ 
 linear in the size of $ONF_{\Sigma}(T)$, 
  such that $ONF_{\Sigma}(T)= \exists \downarrow X. \psi$. 
\end{lemma}
\begin{proof}
We start with $ONF_{\Sigma}(T)$ and a 
$(\Sigma,W)$ oversampled behaviour $\rho'$. 
Let $dom(\rho')=\{1,2,\dots,n\}$.
If there exists a point $i \in dom(\rho')$ marked $act \land a$,  then 
we want to ensure that  all points $j$ in $dom(\rho')$ marked $act$ such that $\tau'_j \in [\tau'_i+l, \tau'_i+u)$ are marked $b$.  This is enforced by the following formula:
\begin{itemize}
 \item $\mathsf{MARK}_b:  \wB[(a\wedge act) \rightarrow \Box_{[l,u)}(act \rightarrow b)]$
\end{itemize}
$\mathsf{MARK}_b$ enforces the direction $act \rightarrow (\past_{[l,u)} (a \wedge act) \rightarrow b)$ of $ONF_{\Sigma}(T)$.
  Marking points of $\rho'$ with $\neg b$ is considerably more involved. 
  We use Lemma \ref{lemmapast} to characterize the points where $\neg \past_{[l,u)}a$ holds,
 and use this to ensure that such points are marked $\neg b$.  
 Recall that by Lemma \ref{lemmapast}, such points can be classified into three regions. 
% Points in regions I and II are handled respectively by formulae 
% $\mathsf{MARK}_{first}$ and 
% $\mathsf{MARK}_{last}$ as in Lemma \ref{past-b1}.
 
 Region I  consists of all those points to the left of $\tau_{first_a}+l$.  
 In any model, these points are described by  the formula   
 $\mathsf{MARK}_{first}=  \wB(\neg a \wedge \neg b) \vee (\neg a\wedge  \neg b) \wU (a \wedge \wB_{[0,l)} \neg b)$\footnote{when $l=0$, $\wB[(\neg a \wedge \neg b) \vee [(\neg a\wedge  \neg b) \wU (a \wedge \neg b)]]$}, which  says that there are no $b$'s in $[0, \tau_{first_a}+l)$.  
  Region II consists of all points in $[\tau_{last_a}+u, \infty)$. In any model, 
  these  points are captured by the formula 
  $\mathsf{MARK}_{last}=\wB(\Box \neg a \rightarrow \Box_{[u,\infty)} \neg b)$, which says that 
  there are no $b$'s in $[\tau_{last_a}+u, \infty)$. 

  Let us now discuss how to mark points lying in region III with $\neg b$. Recall that these are the points 
   in $[\tau_j+u, \tau_k+l)$ for any two consecutive points 
$j,k$ such that $a \in \sigma_j, \sigma_k$, but $a \notin \sigma_h, j < h < k$. 
 Consider $j,k$ as two consecutive points where $a$ holds. If $\tau_k-\tau_j \leq u-l$, then 
 clearly, there are no points in $[\tau_j+u, \tau_k+l)$ to be marked $\neg b$. 
   Assume now that $\tau_k-\tau_j > u-l$. We need to mark exactly the points falling in $[\tau_j+u,\tau_k+l)$ 
  with $\neg b$. It is quite possible that, we dont have the points $g,h$ in $dom(\rho')$ such that 
  $\tau_g=\tau_j+u$ and $\tau_h=\tau_k+l$. Here, we use the idea of oversampled projections, to obtain a 
  behaviour $\rho''$ from $\rho'$, by adding extra points to $dom(\rho')$.   
  Corresponding to every pair $j,k$ of consecutive $a$ points, 
  such that $\tau_k-\tau_j > u-l$, 
  we add  points $x,y$ to $dom(\rho')$, such that $\tau_x=\tau_j+u$ and $\tau_y=\tau_k+l$. We mark 
  these new points with fresh propositions $beg_b$ and $end_b$ respectively. We then say that 
  between $beg_b$ and $end_b$, no $b$ can occur. To pindown the points $x,y$ correctly,  
  we mark the points $j,k $ respectively with fresh propositions $b_s$ and $b_s$. 
  
   To summarize the marking scheme, given a $(\Sigma,W)$-oversampled behaviour $\rho'$ satisfying 
    $ONF_{\Sigma}(T)$, where $T=\wB[b \leftrightarrow \past_{[l, u)}a]$,  
   we construct a $(\Sigma \cup W, X)$-oversampled behaviour $\rho''$ from $\rho'$, such that
    \begin{itemize}
        \item $\rho''$ is obtained by introducing extra  points 
        to $dom(\rho')$. These extra points are 
        related to consecutive $a$ points $j,k \in dom(\rho')$,  when $\tau_k-\tau_j > u-l$. 
     For such $j,k \in dom(\rho')$, we add points $x,y$ to $dom(\rho'')$ such that $\tau_x= \tau_j+u$ and $\tau_y=\tau_k+l$. 
    The fresh propositions used so far, consists of symbols $\{b_s,b_e,beg_b,end_b\} \subseteq X$. 
      \item Symbols $b_s$ and $b_e$ represent the ``start'' and ``end'' positions 
  $j,k$. Thus, $b_s$ holds at a point where $a \wedge act$ is true, and 
  where the next consecutive occurrence of $a$ is  
  $> u-l$ distance apart. Similarly, $b_e$ holds at a point where $a \wedge act$ is true, and 
  where the previous occurrence of $a$ is  
  $> u-l$ distance apart.
       Once we mark $\tau_j$ with $b_s$ and $\tau_k$ with $b_e$, the points at $\tau_j+u$ and $\tau_k+l$   
     are marked $beg_b$ and $end_b$ respectively.  
       Once we have the points $beg_b$ and $end_b$ marked, 
     we assert that between any consecutive pair of $beg_b$ and $end_b$, 
     all points of $\rho'$ are marked $\neg b$. 
%      We assert that 
%    the  $beg_b$'s and $end_b$'s should alternate : that is, there are no two $beg_b$'s 
%    without an $end_b$ in between, and there are no two $end_b$'s 
%    without a $beg_b$ in between. 
      \item We need to make sure that the $beg_b$ and $end_b$ occurring in $\rho''$ are legitimate with respect to $b_s$ and $b_e$: 
That is, there must be no ``free occurrence'' of $beg_b$ and $end_b$. 
Any occurrence of $beg_b$ and $end_b$ should witness $b_s$ and $b_e$ at exactly $u$ and $l$ distance in the past respectively. This can be done adding extra points at all integer timestamps and restricting the free occurrences of $beg_b,end_b$ in every unit interval.
% if we mark positions $\tau_j+u$ with $beg_b$ and $\tau_k+l$ 
% with $end_b$, then we do not want two points $p_1,p_2$ such that 
% $\tau_j+u < \tau_{p_1} < \tau_{p_2} < \tau_k+l$, such that $p_1$ is marked      
%      $end_b$ and $p_2$ is marked $beg_b$. If this happens, then 
%      the interval $[p_1,p_2]$ can be incorrectly marked $b$. 
%    To avoid this situation, we explicitely 
%    mark all points in $(\tau_j+u, \tau_k+l)$ with $\neg beg_b$;
%      this    along with the condition that $beg_b$ and $end_b$ strictly alternate 
%    rule out spurious occurrences of $beg_b,end_b$ 
%    $(\tau_j+u, \tau_k+l)$.  
              \end{itemize}
  Now we write formulae in $\mtl$ that implement the above, 
  which will hold good on the 
   $(\Sigma \cup W, X)$-oversampled behaviour $\rho''$ from $\rho'$.
 \begin{itemize}
 \item Mark $b_s$ and $b_e$ at points $j$ and $k$: The conjunction of the following two formulae 
 is denoted $\mathsf{MARK}_{j,k}$.\\ 
     $\wB (b_s \leftrightarrow (a\wedge act \wedge (act \rightarrow \neg  a) \until_{(u-l,\infty)}(a\wedge act)))$, \\
   $\wB(b_e \leftrightarrow (a \wedge act \wedge (act \rightarrow \neg a) \since (b_s\wedge act)))$\footnote{$\since$ can be removed from $\mathsf{MTL}[\until_I,\since]$
 obtaining equisatisfiable formula in $\mtl$ modulo simple projections \cite{deepak08}, details in Appendix \ref{since-rem}}
    \item Mark $beg_b$ and $end_b$ appropriately at $\tau_j+u$ and $\tau_k+l$ respectively. The conjunction of the following two formulae 
    is denoted $\mathsf{MARK}_{beg,end}$.\\ 
      $\wB(b_s \leftrightarrow 
      (\wF_{[0,u)} \Box \bot \vee [
      \Box_{(u,u+1)}\neg beg_b \wedge \fut_{[u,u+1)} beg_b \wedge \Box_{(u-1,u)}\neg beg_b]))$,\\   
      $\wB(b_e \leftrightarrow 
      (\wF_{[0,l)} \Box \bot \vee [
      \Box_{(l-1,l)}\neg end_b \wedge \fut_{(l-1,l]} end_b \wedge \Box_{(l,l+1)}\neg end_b]))$\footnote{when $l=0$,
      $\wB([b_e \leftrightarrow end_b] \wedge [b_e \rightarrow \neg b])$}   
%        \item $beg_b$, $end_b$ alternate. The conjunction of the following formulae is denoted $\mathsf{MARK}_{alt}$.\\
%   $\wB(beg_b \rightarrow (\neg beg_b \wedge \neg end_b) \until (end_b \vee \Box \bot))$\\
%   $\wB(end_b \rightarrow (\neg end_b \wedge \neg beg_b) \until (beg_b \vee \Box \bot))$
    \item Note that the above formula only asserts where $beg_b$ and $end_b$ should occur. We must assert that all other remaining points $beg_b$ and $end_b$  do not occur. This is done as follows: 
 \begin{itemize}
\item  First mark all  integer timestamps with a fresh proposition $c$. The following  formula is denoted $\mathsf{MARK}_c$. \\
  $c \wedge \wB( c \rightarrow  [\fut_{(0,1)} \Box \bot \vee (\Box_{(0,1)} \neg c
 \wedge \fut_{(0,1]} c )]$
  \item We identify the points between $b_s$ and $b_e$ 
  by uniquely marking the closest integral point before $b_s$ with $c_{b_s}$ and 
  and the closest integral point before $b_e$ with $c_{b_e}$. 
  Recall that $b_s$ and $b_e$ were marked at $\tau_j$ and $\tau_k$; 
  thus, $c_{b_s}$ and $c_{b_e}$ get marked respectively at points $\lfloor \tau_j \rfloor$ 
  and $\lfloor \tau_k \rfloor$.
  We then assert that $beg_b$ can occur at a point $t$ 
 only if there is a $c_{b_s}$ in $(t-u-1, t-u]$.  
  Thus, given that $c_{b_s}$ is marked at 
  $\lfloor \tau_j \rfloor$, $beg_b$ is marked only 
  in $[\lfloor \tau_j \rfloor+u, \lfloor \tau_j \rfloor+u+1)$.
  However, by formula $\mathsf{MARK}_{beg,end}$,
  we disallow $beg_b$ in $(\tau_j+u, \tau_j+u+1)$
  and $(\tau_j+u-1, \tau_j+u)$.
   Thus, we obtain a unique marking for $beg_b$.
    In a similar manner, we obtain a unique marking for $end_b$, given $b_e$.
          The conjunction of the following formulae denoted  $\mathsf{MARK}_{c_b}$ 
  marks $c_{b_s}$ and $c_{b_e}$, and 
  controls the marking of $beg_b$ and $end_b$ correctly:\\
  $\wB[c_{b_s} \leftrightarrow (c \wedge \wF_{[0,1)} b_s)] \wedge 
\wB[c_{b_e} \leftrightarrow (c \wedge \wF_{[0,1)} b_e)]$\\
   $\wB[c \wedge \neg c_{b_s} \rightarrow  \wB_{[u,u+1)} \neg beg_b]$\\  
  $\wB[c \wedge \neg c_{b_e} \rightarrow \wB_{[l,l+1)} \neg end_b]$\\
    Note that these formula do not restrict the behavior of $beg_b$ and $end_b$  in the prefix $[0,u]$. At these timepoints $beg_b$ and $end_b$ should not occur. Here we assert that 
    $\wB_{[0,u)}(\neg beg_b \wedge \neg end_b)$
\end{itemize}
  \item   Now that we have precisely placed $beg_b$ and $end_b$, we can assert at all points of $\rho'$ between $beg_b$ and $end_b$,  
   $\neg b$ holds. This formula is denoted 
  $\mathsf{MARK}_{\neg b}$.\\ 
         $\wB\{beg_b \rightarrow (\neg end_b \wedge (act \rightarrow \neg b)) \wU end_b \}$
   \end{itemize}
   Figure \ref{fig:os} illustrates marking of $\neg beg_b$. 
  
  Let $\mathsf{MARK}=\mathsf{MARK}_b \wedge \mathsf{MARK}_{first} \wedge  \mathsf{MARK}_{c} \wedge \mathsf{MARK}_{last} 
  \wedge \mathsf{MARK}_{j,k} \wedge \mathsf{MARK}_{beg,end} \wedge \mathsf{MARK}_{\neg b} \wedge \mathsf{MARK}_{c_b}.$\footnote{when $l=0$, conjunct $\wB([a \wedge \Box_{[0,u)}\neg a \wedge \fut_{[0,u]}a] \rightarrow \neg b)$ to $\mathsf{MARK}$} 
  Let $\Sigma_i=\Sigma \cup W \cup X$, 
  for $X=\{b_e,b_s,beg_b,end_b,c,c_{b_s},c_{b_e}\}$. Then, $\rho''$ is a $(\Sigma \cup W, X)$-oversampled behaviour 
  such that $\rho'' \models ONF_{\Sigma_i}(\mathsf{MARK})$ iff $\rho' \models ONF_{\Sigma}(T)$. 
  That is, $ONF_{\Sigma}(T)= {\exists \downarrow X}. ONF_{\Sigma_i}(\mathsf{MARK})$. 
   A detailed proof of correctness can be seen in Appendix \ref{os-proof}.
    \end{proof}

\begin{theorem}
\label{theo:main}
For every $\varphi \in \mtlsns$ over $\Sigma$, we can construct $\psi_{fut}$ in $\mtl$ 
over $\Sigma' \supseteq \Sigma$ such that $\varphi = \exists \downarrow X.\psi_{fut}$, $X=\Sigma'-\Sigma$. 
\end{theorem}
\begin{proof}
 Follows from the fact that 
 $\since_{np}$ can be expressed using $\since$ and 
$\past_{np}$\footnote{For instance, we can write $a \since_{[l, r)}b$ as 
$\past_{[l, r)}b ~\wedge ~(a \since b) \land \boxminus_{[0,l)}(a \wedge a \since b)$, for $r=l+1, \infty$} 
 \cite{deepak08}
and elimination of $\since$ 
 \cite{deepak08}, \cite{formats11}.
\end{proof}

By symmetry, using reflection \cite{formats11}, the satisfiability of 
$\mtluns$ can be reduced to the satisfiability of $\mtlsns$. Hence, the satisfiability of 
 $\mtluns$ is also decidable.
%Appendix \ref{ex} illustrates in detail, the elimination of a past modality $\past_{[l,l+1)}a$.

\subsection{Elimination of Past with Simple Projections}
\label{sp-main}
This section is devoted to showing that given any $\varphi \in \mathsf{MTL}[\until_I, \past_{np}]$ built from  $\Sigma$, we can synthesize $\varphi' \in \mtl$ built from $\Sigma'$ such that $\varphi=\exists X. \varphi'$, where 
$X=\Sigma'-\Sigma$. The main steps are similar to the case of oversampling 
projections.
%\subsection{Steps to be followed in Simple Projections}
%\label{steps}
Here are the steps:  
\begin{enumerate}
\item Start with a formula $\varphi \in \mathsf{MTL}[\until_I, \past_{np}]$ built from $\Sigma$, and a timed word $\rho$ over $\Sigma$. 
 Flatten $\varphi$ obtaining $\varphi_{flat}=\psi \wedge \bigwedge_{i=1}^kT_i$. Each $T_i$ 
is a temporal definition of the form $\wB(b_i \leftrightarrow \past_{\langle l, u \rangle}a_i)$, and $\psi \in \mtl$.  
  Let $w_i$ be the fresh witness variable introduced in the temporal definition $T_i$. 
Let $W=\{w_1, \dots,w_n\}$ be the set of all the witness variables.  
\item As discussed in section \ref{sec-flat},  $\varphi=\exists W. ENF_{\Sigma}(\varphi_{flat})$.
\item We now synthesize {\it modulo simple projections}, formulae in $\mtl$ 
equisatisfiable  with  $ENF_{\Sigma}(T_i)$ for $i=1,2, \dots,k$, modulo simple projections.  
  \item Start with $ENF_{\Sigma}(T_1)$, a formula built from $\Sigma \cup W$.
  %, and   interpreted over $\Sigma \cup W$. 
  % and a $(\Sigma,W)$-simple extension $\rho_1$. 
   Let $\Sigma_1=\Sigma \cup W$. 
  We synthesize a formula $\varphi_1 \in \mtl$ built from $\Delta_1=\Sigma \cup W \cup  X_1$
% and interpreted over $\Sigma _1 \cup X_1$ 
   such that 
$ENF_{\Sigma}(T_1)=\exists X_1 \varphi_1$.
  \item  Repeat step 5 for $ENF_{\Sigma}(T_i)$ for all $2 \leq i \leq k$, obtaining 
 formulae $\varphi_i \in \mtl$ built from some $\Delta_i \supseteq \Sigma_1$ such that  
    $ENF_{\Sigma}(T_i)=\exists X_i.\varphi_i$
in each case. The choice of $\Delta_i$ is such that $(\Delta_i - \Sigma_1) \cap (\Delta_j - \Sigma_1)=\emptyset$ for $i \neq j$.
\item Using Lemma \ref{lem:boolclosedequis}, we obtain 
$ ENF_{\Sigma}(\varphi_{flat})=ENF_{\Sigma}(\psi \wedge \bigwedge_{i=1}^kT_i)=
 \exists X.[\psi \wedge \bigwedge_{i=1}^k \varphi_i]$, where $X=\bigcup_{i=1}^k X_i$. Then we get  
 $\varphi=\exists W. ENF_{\Sigma}(\varphi_{flat})=\exists W.[\exists X.(\psi \wedge \bigwedge_{i=1}^k \varphi_i)]$.
\item Steps 1-7 show that $\psi \wedge \bigwedge_{i=1}^k \varphi_i \in \mtl$ is equisatisfiable to $\varphi$ modulo
simple projections. 
  \end{enumerate}
  
%Appendix \ref{steps} runs you through the steps.  
 Lemma \ref{remove-pastinf2} explains how to eliminate temporal definitions 
of the form $\wB[b \leftrightarrow \past_{\langle l, \infty)}(a)]$, while 
Lemma \ref{past-b2} explains how to eliminate temporal definitions of the form 
$\wB[b \leftrightarrow \past_{\langle l, u \rangle}(a)]$.
If all the past modalities involved have unbounded intervals, 
then we get an equivalent formula, as shown by Lemma \ref{remove-pastinf2}.

\begin{lemma}
\label{remove-pastinf2}
Consider the temporal definition $T=\wB[b \leftrightarrow \past_{[l, \infty)}(a)]$ built from $\Sigma \cup W$. 
Then we can synthesize a formula $\psi \in \mtl$ built from $\Sigma \cup W$ equivalent to $T$. 
\end{lemma}
\begin{proof}
It can be shown that $[\wB(\neg a) \vee 
\wB[a\rightarrow \Box_{[l,\infty)}b]] \wedge [\wB(\neg a\wedge  \neg b) \vee (\neg a\wedge  \neg b)  \wU (a \wedge \wB_{[0,l)} \neg b)]$\footnote{when $l=0$,  $[[\wB(\neg a) \vee 
\wB[a\rightarrow \Box_{[l,\infty)}b]] \wedge [\wB(\neg a\wedge  \neg b) \vee (\neg a\wedge  \neg b)  \wU (a \wedge \neg b)]$} is equivalent to 
$T$. Details can be found in Appendix \ref{proofs-5}.
\end{proof}
% Note that above reductions doesn't assume anything regarding monotonicity and finiteness on models.
% Next, consider a past formula of the form $\past_{[l, u)}$.
%Lemma \ref{lemmapast} \gives regions (intervals with end-points in $\R$) such that the formula 
%$\past_{\langle l, u \rangle}a$ cannot hold at 
%a point $i$ of a timed word if and only if $i$ is in that region.
%Figures \ref{case1}--\ref{case4} depict the
%regions where $\past_{[l, u)}$ does not hold (these are denoted by $\neg b$).
%Readers are urged to study these figures.

\begin{lemma}
 \label{past-b2}
 Consider the temporal definition  $T=\wB[b \leftrightarrow \past_{\langle l,u \rangle}(a)]$ 
 built from $\Sigma \cup W$.  We can synthesize a formula $\psi \in \mtl$ built from $\Sigma \cup W \cup X$   such that  $ENF_{\Sigma}(T) = \exists X. \psi$. 
  % Consider a temporal definition . Let $a,b \in \Sigma$. 
% Let $\Sigma' = \{x_0,x_1,y_0,y_1,a_0,a_1,b_{is},b_{ie},b_s,b_e|i\in (u-l,\ldots,l) \}$. 
% Then we can synthesize $\psi \in \mtl$ over $\Sigma \cup \Sigma'$ such that $\varphi = \exists X. \psi$
\end{lemma}
\begin{proof}
We prove the lemma for temporal definitions of the form 
$T=\wB[b \leftrightarrow \past_{[l,u)}(a)]$. The proof can be extended to all kinds 
of intervals $\langle l, u \rangle$. 

Note that $T$ is the conjunction of $C_1 = \wB[b \leftarrow \past_{[l,u)}a]$ and 
$C_2 = \wB[b \rightarrow \past_{[l,u)}a]$. Consider a timed word $\rho$ over $\Sigma \cup W$.    
 $\rho$ satisfies $C_1$ iff, 
for all points $j \in dom(\rho)$, if there exists a point  
 $i \in dom(\rho),$  with $\tau_i \in (\tau_j-u, \tau_j-l]$ and $a \in \sigma_i$, 
 then $b \in \sigma_j$.  
  Clearly, such models $\rho$ are such that 
%\begin{enumerate}
%\item  either there is no $i \in dom(\rho)$ marked $a$ or $b$, or
%\item 
whenever $a \in \sigma_i$, then $b \in \sigma_j$ for all $j \in dom(\rho)$ 
such that $\tau_j \in [\tau_i+l,\tau_i+u)$. 
%\end{enumerate}
Let $\mathsf{MARK}_b= 
\wB[a\rightarrow \Box_{[l,u)}b]$. 
Clearly,  $\rho \models \mathsf{MARK}_b$ iff $\rho \models C_1$.

For a word $\rho$ to satisfy $T$, the above conditions are not enough, since they only characterize points 
in the model where $b$ hold. The formula $\mathsf{MARK}_b \in \mtl$ does not characterize points 
where $b$ should not hold. Models satisfying $\mathsf{MARK}_b$ can allow a point where 
$b$ as well as $\neg   \past_{[l,u)}a$ holds.  
Our next goal is therefore, to find a formula $\mathsf{MARK}_{\neg b} \in \mtl$ which is equisatisfiable to  $C_2$. 
Then $\mathsf{MARK}_b \wedge \mathsf{MARK}_{\neg b}$ is the formula in $\mtl$ that is equi-satisfiable to $T$.

 We use Lemma \ref{lemmapast} to characterize the points where $\neg \past_{[l,u)}a$ holds,
 and use this to ensure that such points are marked $\neg b$.  
 Recall that by Lemma \ref{lemmapast}, such points can be classified into three regions. 
 Points lying in Regions I,II are handled by the formulae 
 $\mathsf{MARK}_{first},\mathsf{MARK}_{last}$ given in Lemma \ref{past-b1}.
 So far, we have the conjunction $\mathsf{MARK}_{first} \wedge \mathsf{MARK}_{last} \wedge 
\mathsf{MARK}_{b}$ of formulae in $\mtl$.

Region III consists of all points in $[\tau_j+u, \tau_k+l)$ for any pair of consecutive ``$a$'' points 
$j,k$ ($a \in \sigma_j, \sigma_k$ and for all $j<h<k$, $a \notin \sigma_h$). The difficulty in marking points 
in $[\tau_j+u, \tau_k+l)$ with $\neg b$ is :
\begin{enumerate}
\item    Points $p_1, p_2 \in dom(\rho)$ 
with $\tau_{p_1}=\tau_j+u$ and $\tau_{p_2}=\tau_k+l$ may not be present in $dom(\rho)$;
\item  The length of the region $[\tau_j+u, \tau_k+l)$ may not be an integer. If it were, 
we can pin down points in  $[\tau_j+u, \tau_k+l)$ by anchoring at points $j, k$
since $l,u$ are integers.
\end{enumerate}
 Unless we can pin down these points or mark this region uniquely, we may end up 
marking lesser points than necessary with $\neg b$ or 
may mark a point already marked $b$ with $\neg b$, giving rise to inconsistencies. 
 The rest of the proof is devoted to 
showing how we can indeed pin down the set of points between $\tau_j+u$ and $\tau_k+l$.

 Since we may not have the points $\tau_j+u$ and $\tau_k+l$, we try to get points as close as possible 
 to $\tau_j+u$ and $\tau_k+l$, 
 by considering an over approximation of the interval $[\tau_j+u, \tau_k+l)$. 
  The idea is to express $[\tau_j+u, \tau_k+l)$ as the intersection 
 of two intervals $I^1_{j,k}$ and $I^2_{j,k}$,  both having integer length, and such that 
 it is possible to pin down $I^1_{j,k}$ and $I^2_{j,k}$. 
      For this, we consider the intervals 
 $I^1_{j,k}=[\tau_k+l-d, \tau_k+l)$ and $I_{j,k}^2=[\tau_j+u, \tau_j+u+d)$ where $d=\lceil \tau_k - \tau_j \rceil + (l-u)$. 
 Note that $d$ is the closest integer that is larger than the actual duration of the interval 
  $[\tau_j+u, \tau_k+l)$. Also, $\tau_k+l-d \leq \tau_j+u$ and $\tau_k+l \leq \tau_j+u+d$. 
  Hence, $[\tau_j+u, \tau_k+l) \subseteq I^1_{j,k} \cap I^2_{j,k}$.  
  We now pin down points in the intersection $I^1_{j,k} \cap I^2_{j,k}$ 
 and mark them $\neg b$. Towards getting the points in the intersection, we 
 allow marking points $i \in dom(\rho)$ with fresh witness propositions, obtaining 
 from $\rho$, a simple extension $\rho'$.
 
 \begin{figure*}[ht]
  \begin{center}
  \begin{picture}(45,8)(20,-7)
  \thicklines
  \drawline[AHnb=0,ATnb=0](-15,0)(105,0)
  \drawline[AHnb=0,ATnb=0](-15,-2)(-15,2)
  \drawline[AHnb=0,ATnb=0](105,-2)(105,2)
  \put(-5,4){\tiny $a\wedge a_0$}
  \put (-5,-4){$\tau_j$}
   \drawline[AHnb=0,ATnb=0](-5,-2)(-5,2)
  \put(5,4){\tiny $a \wedge a_0$}
  \put (5,-4){$\tau_k$}
   \drawline[AHnb=0,ATnb=0](5,-2)(5,2)
  %\node[Nframe=n,fillcolor=white](A)(-16,6){}
  %\node[Nframe=n,fillcolor=white](B)(16,6){}
  \drawline[AHnb=0,ATnb=0](-5,-6.5)(5,-6.5)
  \drawline[AHnb=0,ATnb=0](-5,-5)(-5,-8)
   \drawline[AHnb=0,ATnb=0](5,-5)(5,-8)
  \drawline[AHnb=0,dash={0.25}0	](55.5,-6)(70.5,-6)
  \put(55,-6.5){[}
  \put(55,-8,3){\tiny{$\tau_j+l$}}
  \put(70,-8.3){\tiny{$\tau_j+u$}}
  \put(70,-6.5){)}
  \drawline[AHnb=0,dash={0.25}0	](55.5,0)(55.5,-6.5)
  \drawline[AHnb=0,dash={0.25}0	](70.8,0)(70.8,-6.5)
  \put(65,-10.5){[}
  \put(80,-10.5){)}
   \put(65,-12.5){\tiny{$\tau_k+l$}}
   \put(80,-12.5){\tiny{$\tau_k+u$}}
   \drawline[AHnb=0,ATnb=0	](65.5,0)(65.5,-11)
   \drawline[AHnb=0,ATnb=0	](80.8,0)(80.8,-11)
   \drawline[AHnb=0,dash={0.25}0	](65.5,-10.2)(80.5,-10.2)
  \put(-5,-9){\tiny{$<u-l$}}
 % \put(101,-4){$t+l+2$}
 % \put(104,4){$c$}
 % \drawline[AHnb=0,ATnb=0](0,-2)(0,2)
 % \put(-3,-4){$t+1$}
 % \put(-1,4){$c$}
 %\drawline[AHnb=0,ATnb=0](16,-2)(16,2)
 % \put(13,-4){$t+2$}
 % \put(15,4){$c$}
 % \drawline[AHnb=0,ATnb=0](85,-2)(85,2)
 % \put(81,-4){$t+l+1$}
 % \put(83.5,4){$c$}
 % \drawline[AHnb=0,dash={1.5}0](85,-6)(103.5,-6.5)
 % \put(84.5,-6.5){[}
 % \put(104,-6.5){)}
 % \put(90,-8.5){$\neg b$}
 % \drawline[AHnb=0,ATnb=0](65,-2)(65,2)
 % \put(61,-4){$t+l$}
 % \put(64,4){$c$}
  \end{picture}
  \caption{$\tau_k-\tau_j \leq u-l$}
  \label{fig:case1}
  \end{center}
  \end{figure*}
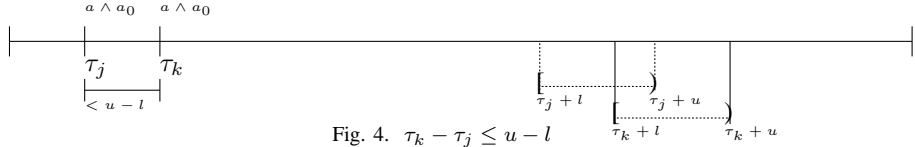

  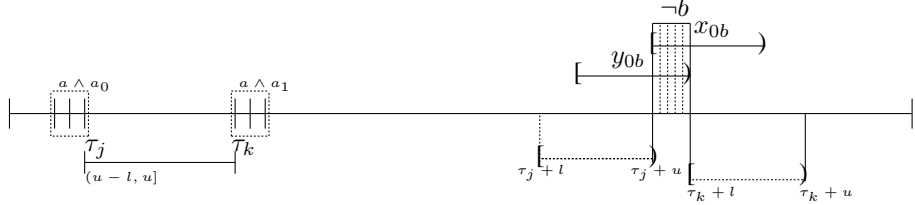
\begin{figure*}[th]
  \begin{center}
  \begin{picture}(45,18)(20,-9)
  \thicklines
   \drawline[AHnb=0,ATnb=0](-15,0)(105,0)
   \drawline[AHnb=0,ATnb=0](-15,-2)(-15,2)
   \drawline[AHnb=0,ATnb=0](105,-2)(105,2)
   
   \put(-8.5,3.5){\tiny $a \wedge a_0$}
   \put (-5,-5){$\tau_j$}
    \drawline[AHnb=0,ATnb=0](-5,-2)(-5,2)
    \drawline[AHnb=0,ATnb=0](-7,-2)(-7,2)
    \drawline[AHnb=0,ATnb=0](-9,-2)(-9,2)
   \put(15.5,3.5){\tiny $a \wedge a_1$}
   \put (14.5,-5){$\tau_k$}
    \drawline[AHnb=0,ATnb=0](15,-2)(15,2)
	\drawline[AHnb=0,ATnb=0](17,-2)(17,2)
	\drawline[AHnb=0,ATnb=0](19,-2)(19,2)
   %\node[Nframe=n,fillcolor=white](A)(-16,6){}
   %\node[Nframe=n,fillcolor=white](B)(16,6){}
   \drawline[AHnb=0,ATnb=0](-5,-6.5)(15,-6.5)
   \drawline[AHnb=0,ATnb=0](-5,-5)(-5,-8)
    \drawline[AHnb=0,ATnb=0](15,-5)(15,-8)
 \drawline[AHnb=0,dash={0.25}0](19.5,3)(14.5,3)
\drawline[AHnb=0,dash={0.25}0](19.5,3)(19.5,-3)

\drawline[AHnb=0,dash={0.25}0](14.5,3)(14.5,-3)

\drawline[AHnb=0,dash={0.25}0](19.5,-3)(14.5,-3)

   \drawline[AHnb=0,dash={0.25}0](55.5,-6)(70.5,-6)

 \drawline[AHnb=0,dash={0.25}0](-9.5,3)(-4.5,3)
\drawline[AHnb=0,dash={0.25}0](-9.5,3)(-9.5,-3)

\drawline[AHnb=0,dash={0.25}0](-4.5,3)(-4.5,-3)

\drawline[AHnb=0,dash={0.25}0](-9.5,-3)(-4.5,-3)

   \drawline[AHnb=0,dash={0.25}0](55.5,-6)(70.5,-6)

    \put(-5,-9){\tiny{$ (u-l,u]$}}

   \put(55,-6.5){[}
   \put(52.5,-8){\tiny{$\tau_j+l$}} %55 was the actual position
   \put(67.5,-8){\tiny{$\tau_j+u$}} %70 was the actual position
   \put(70,-6.5){)}
   \drawline[AHnb=0,dash={0.25}0	](55.5,0)(55.5,-6.5)
   \drawline[AHnb=0,ATnb=0	](70.5,12)(70.5,-6.5)
   \drawline[AHnb=0,ATnb=0](70.5,12)(75.5,12)
   \drawline[AHnb=0,dash={0.25}0](71.5,12)(71.5,0)
   \drawline[AHnb=0,dash={0.25}0](72.5,12)(72.5,0)
   \drawline[AHnb=0,dash={0.25}0](73.5,12)(73.5,0)
   \drawline[AHnb=0,dash={0.25}0](74.5,12)(74.5,0)
   \put(75,-9){[}
   \put(90,-9){)}
    \put(75,-11){\tiny{$\tau_k+l$}}
    \put(90,-11){\tiny{$\tau_k+u$}}
    \drawline[AHnb=0,ATnb=0	](75.5,12)(75.5,-9)
    \drawline[AHnb=0,ATnb=0	](90.8,0)(90.8,-9) 
%    \drawline [AHnb=0,ATnb=0	](70.5,4) (75.5,4)
   \drawline[AHnb=0,dash={0.25}0	](75.5,-8.8)(90.5,-8.8)
    \put(71.5,13){$\neg b$}
\put (65,6.75){$y_{0b}$}
\put(60,4.5){[}
\put(74.5,4.5){)}
\put (76,10.75){$x_{0b}$}
\put(70.,8.5){[}
\put(84.5,8.5){)}

         \drawline [AHnb=0,ATnb=0	](60.5,5) (75.5,5)
         \drawline [AHnb=0,ATnb=0	](70.5,9) (85.5,9)
 
  \end{picture}
  \caption{$\tau_k-\tau_j \in (u-l,u]$. The shaded region indicates  $x_{0b} \wedge y_{0b}$. This  region is marked $\neg b$}
  \label{fig:case2}
  \end{center}
  \end{figure*}

  \begin{figure*}[th]
  \begin{center}
  \begin{picture}(46,22)(20,-16)
   \thicklines
    \drawline[AHnb=0,ATnb=0](-15,0)(105,0)
    \drawline[AHnb=0,ATnb=0](-15,-2)(-15,2)
    \drawline[AHnb=0,ATnb=0](105,-2)(105,2)
    \put(-6.5,2){\tiny $a\wedge a_0$}
    \put (-5,-4){$\tau_j$}
    
     \drawline[AHnb=0,ATnb=0](-5,2)(-5,-9)
    \put(48,2){\tiny $a\wedge a_1$}
    \put (50,-4){$\tau_k$}
     \drawline[AHnb=0,ATnb=0](50,2)(50,-9)
    %\node[Nframe=n,fillcolor=white](A)(-16,6){}
    %\node[Nframe=n,fillcolor=white](B)(16,6){}
    \drawline[AHnb=0,ATnb=0](-5,-9)(50,-9)
    \drawline[AHnb=0,ATnb=0](-5,-5)(-5,-8)
     \drawline[AHnb=0,ATnb=0](50,-5)(50,-8)
    \drawline[AHnb=0,dash={0.25}0](25.5,-5)(40.5,-5)
     \put(5,-12){\tiny{$ > u$}}
    \put(25,-5.5){[}
    \put(25,-7.5){\tiny{$\tau_j+l$}}
    \put(40,-7.5){\tiny{$\tau_j+u$}}
    \put(40,-5.5){)}
    \drawline[AHnb=0,dash={0.25}0	](25.5,0)(25.5,-5.5)
    \drawline[AHnb=0,dash={0.25}0	](40.8,6.5)(40.8,-5.5)
    \drawline[AHnb=0,ATnb=0	](40.5,6.5)(80.5,6.5)
    \put(80,-11){[}
    \put(95,-11){)}
     \put(80,-13){\tiny{$\tau_k+l$}}
     \put(95,-13){\tiny{$\tau_k+u$}}
     \drawline[AHnb=0,ATnb=0	](80.5,6.5)(80.5,-11)
     \drawline[AHnb=0,ATnb=0	](95.8,0)(95.8,-11) 

    \drawline[AHnb=0,dash={0.25}0	](80.5,-10.8)(95.5,-10.8)
     \put(60.5,7.5){$\neg b$}
  
  \end{picture}
  \caption{$\tau_k-\tau_j > u$}
  \label{fig:case3-sp}
  \end{center}
  \end{figure*}
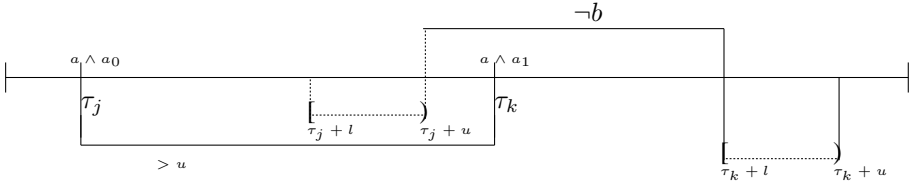

 In the following, we 
 explain the choice of these propositions, the marking scheme 
 to obtain $\rho'$, 
 and 
 formulae in $\mathsf{MTL}[\until_I,\since]$\footnote{$\since$ can be removed from $\mathsf{MTL}[\until_I,\since]$
 obtaining equisatisfiable formula in $\mtl$ modulo simple projections \cite{deepak08}, details in Appendix \ref{since-rem}}
  which enforce these markings.

 \noindent{\underline{Case 1}}:  If $\tau_k-\tau_j \le u-l$ for consecutive points $j,k$ with $a \in \sigma_j,\sigma_k$. 
  Then $[\tau_j+u, \tau_k+l)$ is the empty interval and 
  $d=\lceil \tau_k - \tau_j \rceil + (l-u) \leq 0$ 
    and hence no action need to be taken. Figure \ref{fig:case1} illustrates this case.\\
\noindent{\underline{Case 2}}:  If $\tau_k-\tau_j \in (u-l,u]$. Then the interval $[\tau_j+u, \tau_k+l)$ is non-empty, and 
$1 \leq d=\lceil \tau_k - \tau_j \rceil + (l-u) \leq l$. 
\begin{enumerate}
\item We introduce two propositions $a_0,a_1$ that marks all 
     positions $i \in dom(\rho)$ such that $a \in \sigma_i$ with a unique 
     element from $\{a_0,a_1\}$.
      The position $first_a$ is marked $a_0$;
     if consecutive $a$'s are at a distance $> u-l$, then they are marked by exactly one of $a_i$ and $a_{1-i}$
         respectively, for $i \in \{0,1\}$ such that they alternate; if consecutive $a$'s are at a distance $\leq u-l$, they are both marked  with exactly the same $a_i$, $i \in \{0,1\}$. 
                  A consecutive $a_i,a_{1-i}$ pair ``flags'' 
         attention : they play a role, in marking some interval with $\neg b$.
             The conjunction 
            of the 
            following formulae, denoted $\mathsf{MARK}_{a}$ implements these: 
\begin{enumerate}
\item 
%$\forall i \in dom(\rho)$, 
%$a \in \sigma_i \rightarrow (a_0 \in \sigma_i \vee a_1 \in \sigma_i)$ :   
 $\wB((a_0 \vee a_1) \leftrightarrow a) \wedge \wB(\neg a_0 \vee \neg a_1)$ 
%\item  
%$\forall i \in dom(\rho)$,   $\neg[\{a_0,a_1\} \subseteq \sigma_i]$ : 
 \item %$(i=first_a) \rightarrow  (a_0 \in \sigma_i)$ : 
 $\neg a \wU (a\wedge a_0)$
\item  %$a_0$ and $a_1$  alternate only if the time diff is long enough:\\ 
$\bigwedge_{i\in\{0,1\}} \wB[F_1 \wedge F_2]$ where \\
$F_1: (a_i \wedge \Box_{[0,u-l]} \neg a) \rightarrow \Box \neg a \vee 
(\neg a \until (a \wedge a_{1-i}))$\footnote{Note that points $j,k$ with consecutive $a$'s, such that $\tau_k-\tau_j > u$ also are marked by $a_i,a_{1-i}$},\\
 $F_2: (a_i \wedge \fut_{[0,u-l]}a ) \rightarrow  \neg a \until (a \wedge a_i)$.
%  where $i'$ denotes the 1's complement of $i$.  
 \end{enumerate}   
 \item To easily identify the intervals 
 $I_{j,k}^1$ and $I_{j,k}^2$, we mark the points $j,k \in dom(\rho)$ 
 with propositions $beg_{db}$ and $end_{db}$. 
The $d$ in suffix is  $d= \lceil \tau_k - \tau_j \rceil + (l-u)$, the  $b$ in suffix is the witness 
proposition for $\past_{[l,u)}a$, while $beg$, $end$ signify the beginning and end 
of respective consecutive $a$ positions.   
To correctly get the $d$, we need to check the closest unit interval
corresponding to $\tau_k-\tau_j$ : for instance, if $\tau_k-\tau_j=(u-l) +0.4$, then 
we know $\tau_k-\tau_j \in (u-l, u-l+1]$. In this case, 
$\lceil \tau_k - \tau_j \rceil=u-l+1$, and hence, $d=1$. 
We need to do this for all the $l-1$ possibilities :  
 $\tau_k-\tau_j \in (t,t+1]$, where $t \in \{u-l,\ldots,u-1\}$. In each case, 
 the symbols marking the respective consecutive $a$'s 
 will be $beg_{t+1+l-u~b}$ and $end_{t+1+l-u~b}$, where
 $t+1=\lceil \tau_k - \tau_j \rceil$.
 
  To summarize, we introduce propositions $\{beg_{db},end_{db} \mid 1 \leq d \leq l\}$
  to mark two consecutive $a$'s that are at a distance in $(u-l,u]$. 
  The $d$ in the suffix is  the closest integer $\geq$ 
  the duration of the interval $[\tau_j+u, \tau_k+l)$.  
  This is used in the next step to mark correctly the intervals $I_{j,k}^1$ and 
  $I_{j,k}^2$, both of which have duration $d$ : 
  Identifying points $j,k$ with $beg_{db}$ and $end_{db}$, 
     $I^1_{j,k}$ is the interval $[\tau_{end_{db}}+l-d,\tau_{end_{db}}+l)$
while $I^2_{j,k}$ is the interval $[\tau_{beg_{db}}+u,\tau_{beg_{db}}+d+u)$.   
   Note that a unique value of $d$ will only satisfy formula 2(a) below: 
 that value is $d=\lceil \tau_k-\tau_j \rceil +(l-u)=t+1+l-u$.

       The following formulae implement this idea by ensuring that  $beg_{db}$ and $end_{db}$ 
 indeed correspond to consecutive points $j,k$ with $a \in \sigma_j,\sigma_k$. 
 For $t \in \{u-l, \dots, u-1\}$, and $d \in \{1, \dots, l\}$,  
  \begin{enumerate}
 \item $\wB(beg_{t+1+l-u~b} \leftrightarrow(a \wedge (\neg a \until_{(t,t+1]} a)))$.
\item  $\wB(end_{db}\leftrightarrow(a \wedge (\neg a \since~ beg_{db})))$.
 \end{enumerate}
  Let $\mathsf{MARK}_{beg,end,d}$ be the conjunction of the above formulae.
  \item The propositions $beg_{db}$ and $end_{db}$ now help us in identifying the relevant 
 points in the intersection of $I_{j,k}^1$ and $I_{j,k}^2$ as follows: 
 Recall that points $j,k$ marked with $beg_{db},end_{db}$ are also marked with one of $a_0,a_1$
 such that  $\{beg_{db},a_i\} \subseteq \sigma_j$ iff $\{end_{db},a_{1-i}\} \subseteq \sigma_k$. 
 We now identify the points in $I_{j,k}^1=[\tau_{end_{db}}+l-d,\tau_{end_{db}}+l)$ by marking them 
 with a proposition $y_{cb}$ iff $a_{1-c} \in \sigma_k$. 
  Likewise, all the points in $I_{j,k}^2=[\tau_{beg_{db}}+u,\tau_{beg_{db}}+d+u)$ are  marked with a proposition $x_{cb}$ iff $a_c \in \sigma_j$. 
% Recall that the durations of $I^1_{j,k}$ and $I^2_{j,k}$ 
% are both $d$; this $d$ is remembered in  $beg_{db}$ and $end_{db}$. 
% This helps in picking the duration $d$ of intervals to be marked $x_c$ and $y_c$. 
  It can be observed now that points in $I^1_{j,k} \cap I^2_{j,k}$ will be marked with both $x_{cb},y_{cb}$. Such points are marked $\neg b$. Figure \ref{fig:case2} illustrates this.
  %{\it To summarize, we introduce propositions $\{x_0,y_0,x_1,y_1\}$. 
%We mark points of $I^1_{j,k}$ with $y_c$ iff $\{a_{1-c},end_{db}\} \subseteq \sigma_k$ and 
%mark points of $I^2_{j,k}$ with $x_c$ iff $\{a_c,beg_{db}\} \subseteq \sigma_j$.
%Points marked both $x_c,y_c$ are marked $\neg b$. }
 This is implemented by 
 the conjunction of the following formulae, denoted $\mathsf{MARK}_{x,y,c}$:  
\begin{enumerate}
\item $\bigwedge_{c\in\{0,1\}}\wB((beg_{db} \wedge a_c)\rightarrow \Box_{[u,u+d)} x_{cb})$
\item $\bigwedge_{c\in\{0,1\}}\wB((end_{db} \wedge a_c)\rightarrow \wB_{[l-d,l)} y_{1-c~b}))$
\end{enumerate}
\item Let $\mathsf{MARK}_{\neg b,c}$ denote  $\wB((x_{cb} \wedge y_{cb}) \rightarrow \neg b)$, $c \in \{0,1\}$.  
\end{enumerate} 
 
\noindent{\it {\underline{Case 2 Summary}}}:
 We mark consecutive 
  points $j,k$ having $a$  that are apart by a distance in $(u-l,u]$
 with $a_c,a_{1-c}$, $c \in \{0,1\}$, and with $beg_{db}, end_{db}$ respectively, where 
 $d$ is the closest integer that is $\geq \lceil \tau_k-\tau_j \rceil+l-u$.
 The bit $c \in \{0,1\}$ and the value $d$ help in marking 
 all points in $[\tau_k+l-d,\tau_k+l)$ 
 with $y_{cb}$ and all points in $[\tau_j+u, \tau_j+u+d)$ with $x_{cb}$. 
 Points marked both $x_{cb},y_{cb}$ are then marked $\neg b$.

 \noindent{\underline {Case 3}}: $\tau_k - \tau_j \in (u, \infty)$. Then again, 
 $[\tau_j+u, \tau_k+l)$ is non-empty.\footnote{If $l=0$, case 2 gives an empty interval. 
 Case 3 deals with $>u$ distance. For $a$'s which are $u$ apart, 
 we add the formula $(a \wedge \Box_{[0,u)} \neg a \wedge \fut_{[0,u]}a) \rightarrow 
 \fut_{[0,u]} \neg b$} 
   Then  
 $d=\lceil \tau_k - \tau_j \rceil + (l-u) > l$. Figure \ref{fig:case3-sp} illustrates this case.
\begin{enumerate}
\item We introduce propositions 
  $\{b^1_{\infty}, b^2_{\infty}\}$ to mark consecutive 
  $a$'s that are more than $u$ distance apart. 
  We assert that $\neg b$ holds
  in the $[0,l)$ future of $b^2_{\infty}$;
  also $\neg b$ holds at all points 
  that are at a distance $\geq u$ from $b^1_{\infty}$
 and that lie before $b^2_{\infty}$.
  We first mark such consecutive points $j,k$ with propositions $b^1_{\infty}$ and $b^2_{\infty}$. 
  Let $\mathsf{MARK}_{succ,\infty}$ be the conjunction of the following formulae: 
 \begin{enumerate}
 \item $\wB(b^1_{\infty} \leftrightarrow (a \wedge \neg a \until_{(u,\infty)}a))$
 \item $\wB(b^2_{\infty} \leftrightarrow (a \wedge \neg a \since b^1_{\infty}))$
  \end{enumerate}
\item Next we assert that points in 
 $(\tau_j+u, \tau_k]$ and in $[\tau_k, \tau_k+l)$ be marked $\neg b$.
   %\item 
% There might be some points lying in $(\tau_j+u, \tau_k] \subseteq [\tau_j, \tau_k]$
% which should be marked $\neg b$. 
 % We have to assert that $\neg b$ holds 
 %in $(\tau_j+u, \tau_k]$ and in $[\tau_k, \tau_k+l)$. Note that there might be some $b$'s 
 %in $[\tau_j, \tau_j+u]$.  
% We thus have to say that, all points beyond the last $b$ (if it exists) 
 %  between $b^1_{\infty}$ and $b^2_{\infty}$ must be marked $\neg b$. 
 %  Also, there cannot be any $b$  after $b^2_{\infty}$ 
  % for the duration $[0,l)$. 
   This is implemented by the conjunction of the following formulae, denoted      
   $\mathsf{MARK}_{\neg b,\infty}$: 
 \begin{enumerate}
 \item $\wB((b^1_{\infty} \wedge \wF_{[0,u)}b)\rightarrow (\wF_{[0,u)}(b \wedge \neg b \until b^2_{\infty}))$
  \item $\wB((b^1_{\infty} \wedge \wB_{[0,u)} \neg b)\rightarrow (\neg b \wedge \neg b \until b^2_{\infty}))$   
 \item $\wB(b^2_{\infty} \rightarrow \wB_{[0,l)}\neg b)$
  \end{enumerate}
\end{enumerate}
 
  \begin{figure*}[th]
   \begin{center}
   \begin{picture}(45,9)(20,-6)
   \thicklines
    \drawline[AHnb=0,ATnb=0](-15,0)(105,0)
    \drawline[AHnb=0,ATnb=0](-15,-2)(-15,2)
    \drawline[AHnb=0,ATnb=0](105,-2)(105,2)
    \put(-15,4){$a$}
    \put (-15.5,-4){\tiny $3.1$}
    \put(5,4){$a$}
    \put (4.5,-4){\tiny $4.8$}

     \drawline[AHnb=0,ATnb=0](-15,-2)(-15,2)
    \drawline[AHnb=0,ATnb=0](5,-2)(5,2)
    
    \put(25,4){$a$}
    \put (24.5,-4){\tiny $5.9$}
     \drawline[AHnb=0,ATnb=0](25,-2)(25,2)
    %\node[Nframe=n,fillcolor=white](A)(-16,6){}
    %\node[Nframe=n,fillcolor=white](B)(16,6){}
%    \drawline[AHnb=0,ATnb=0](-5,-6.5)(15,-6.5)
%    \drawline[AHnb=0,ATnb=0](-5,-5)(-5,-8)
%     \drawline[AHnb=0,ATnb=0](15,-5)(15,-8)
%    \drawline[AHnb=0,dash={0.25}0	](55.5,-6)(70.5,-6)
%     \put(-5,-9){\tiny{$d\in[u-l,u)$}}
%    \put(55,-6.5){[}
%    \put(52.5,-8){\tiny{$\tau_j+l$}} %55 was the actual position
%    \put(67.5,-8){\tiny{$\tau_j+u$}} %70 was the actual position
%    \put(70,-6.5){)}
    \drawline[AHnb=0,ATnb=0	](35.5,2)(35.5,-5)
\put(35.5,-6.5){\tiny$9.1$}
    \drawline[AHnb=0,ATnb=0	](45.5,2)(45.5,-3)
    \drawline[AHnb=0,ATnb=0	](55.5,2)(55.5,-3)
    \drawline[AHnb=0,ATnb=0	](65.5,2)(65.5,-8.5)
\put (64.5,-9){\tiny$11.8$} 
    \drawline[AHnb=0,ATnb=0	](75.5,2)(75.5,-3)
\put(85,-3){\tiny$12.9$}
\put (52.5,6.5){\tiny $10.8$}
\put (74.5,8.5){\tiny $11.9$}
    \drawline[AHnb=0,ATnb=0	](85.5,2)(85.5,-2)
%    \drawline[AHnb=0,ATnb=0	](70.5,12)(70.5,-6.5)
%    \drawline[AHnb=0,ATnb=0](70.5,12)(75.5,12)
%    \drawline[AHnb=0,dash={0.25}0](71.5,12)(71.5,0)
%    \drawline[AHnb=0,dash={0.25}0](72.5,12)(72.5,0)
%    \drawline[AHnb=0,dash={0.25}0](73.5,12)(73.5,0)
%    \drawline[AHnb=0,dash={0.25}0](74.5,12)(74.5,0)
%    \put(75,-9){[}
%    \put(90,-9){)}
%     \put(75,-11){\tiny{$\tau_k+l$}}
%     \put(90,-11){\tiny{$\tau_k+u$}}
%     \drawline[AHnb=0,ATnb=0	](75.5,12)(75.5,-9)
%     \drawline[AHnb=0,ATnb=0	](90.8,0)(90.8,-9) 
% %    \drawline [AHnb=0,ATnb=0	](70.5,4) (75.5,4)
%    \drawline[AHnb=0,dash={0.25}0	](75.5,-8.8)(90.5,-8.8)
%     \put(71.5,13){$\neg b$}
% \put (65,6.75){$y$}
% \put(60,4.5){[}
% \put(74.5,4.5){)}
% \put (76,10.75){$x$}
% \put(70.,8.5){[}
% \put(84.5,8.5){)}

          \drawline [AHnb=0,dash={0.25}0	](45.5,-5) (63,-5)
	\drawline (45.5,-5)(45.5,-5.5)
	\put (44.5,-6.5){\tiny $10.1$} 
          \drawline [AHnb=0,ATnb=0	](45.5,-5) (45.5,6)
          \drawline [AHnb=0,dash={0.25}0	](47.5,-5) (47.5,6)
          \drawline [AHnb=0,dash={0.25}0	](49.5,-5) (49.5,6)
          \drawline [AHnb=0,dash={0.25}0	](51.5,-5) (51.5,6)
          \drawline [AHnb=0,dash={0.25}0	](53.5,-5) (53.5,6)
          \drawline [AHnb=0,dash={0.25}0	](63,0) (63,-5)
          \put(49.5,3){$\neg b$}
          \put(69.5,3){$\neg b$}
          \drawline [AHnb=0,ATnb=0	](55.5,6) (36.5,6)
          \drawline [AHnb=0,ATnb=0	](55.5,6) (55.5,-5)
          \drawline [AHnb=0,ATnb=0	](36.5,-2) (36.5,6)
          \put(36,-4){\tiny $ $}
                    \drawline [AHnb=0,dash={0.25}0	](65.5,-7) (83,-7)
                    \put(75,-9.5){$x_b$}
                    \put(55,-7.5){$x_b$}
                    \put(45,8.5){$y_b$}
                    \put(65,10.5){$y_b$}
                    \drawline [AHnb=0,ATnb=0	](65.5,-7.5) (65.5,8.5)
                    \drawline [AHnb=0,dash={0.25}0	](83,0) (83,-7)
                    \drawline [AHnb=0,ATnb=0	](75.5,8) (57,8)
                    \drawline [AHnb=0,ATnb=0	](75.5,8.5) (75.5,-7.5)
                    \drawline [AHnb=0,ATnb=0	](75.5,8) (57,8)
                    \drawline [AHnb=0,ATnb=0	](57,-2) (57,8)
                    \drawline[AHnb=0,dash={0.25}0](58.5,8)(58.5,0)
                        \drawline[AHnb=0,dash={0.25}0](58.5,8)(58.5,0)
                        \drawline[AHnb=0,dash={0.25}0](59.5,8)(59.5,0)
                        \drawline[AHnb=0,dash={0.25}0](60.5,8)(60.5,0)
                        \drawline[AHnb=0,dash={0.25}0](61.5,8)(61.5,0)
                        \drawline[AHnb=0,dash={0.25}0](57,1)(63,1)
                        \drawline[AHnb=0,dash={0.25}0](57,2)(63,2)
                        \drawline[AHnb=0,dash={0.25}0](57,3)(63,3)
                        \drawline[AHnb=0,dash={0.25}0](57,4)(63,4)
                        \drawline[AHnb=0,dash={0.25}0](57,5)(63,5)
                        \drawline[AHnb=0,dash={0.25}0](57,6)(63,6)
                        \drawline[AHnb=0,dash={0.25}0](57,7)(63,7)                        
                    \put (56.5,-4){\tiny $10.9$}
  					\drawline [AHnb=0,ATnb=0	](63,-2) (63,8)
  					\put (62,-4){\tiny $11.1$}
  					\drawline [AHnb=0,dash={0.25}0	](73.5,8) (73.5,-7)
  					\drawline [AHnb=0,dash={0.25}0	](71.5,8) (71.5,-7)
  					\drawline [AHnb=0,dash={0.25}0	](67.5,8) (67.5,-7)
  					\drawline [AHnb=0,dash={0.25}0	](69.5,8) (69.5,-7)
%  					\drawline [AHnb=0,dash={0.25}0	](71.5,8) (71.5,-7)
%  					\drawline [AHnb=0,dash={0.25}0	](71.5,8) (71.5,-7)
%  					\drawline [AHnb=0,dash={0.25}0	](71.5,8) (71.5,-7)
  					
   \end{picture}
   \caption{Erroneous intersection: $[l,u)=[6,7), j=3.1,k=4.8,h=5.9$.}
   \label{fig:err}
   \end{center}
   \end{figure*}
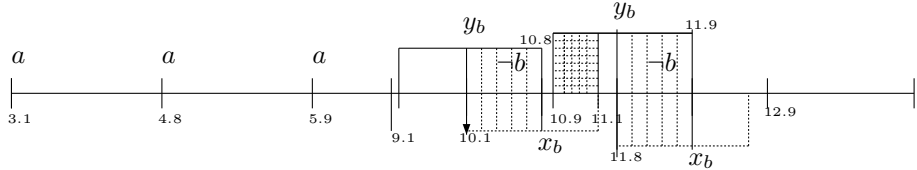

\noindent{\it {\underline{Purpose of Extra Propositions}}}: 
The extra propositions introduced are $X=\{a_0,a_1,x_{0b},x_{1b},y_{0b},y_{1b},b^1_{\infty},b^2_{\infty}\} \cup \{beg_{db},end_{db} \mid 1 \leq d \leq l\}$.
\begin{enumerate}
\item   
First of all, $a_0,a_1$ are chosen 
to enable marking points in $I_{j,k}^1,I_{j,k}^2$ with $x_{0b}, y_{0b}$ or $x_{1b},y_{1b}$, depending on 
whether the point $j$ was marked $a_0$ or $a_1$.  
Consider three consecutive points $j,k,h$ where $a$ holds, with $\tau_k-\tau_j, \tau_h-\tau_k \in [u-l,u]$. 
 Clearly, we are looking at points in $I^1_{j,k},I^2_{j,k}$ and 
$I^1_{k,h},I^2_{k,h}$. If we just had $x_b,y_b$ to mark these intervals,
then we get points in $I_{j,k}^1,I_{k,h}^1$ marked with $y_b$, and 
points in $I_{j,k}^2,I_{k,h}^2$ marked $x_b$. There is a possibility as illustrated by the example below, that 
points marked $x_b$ in $I_{j,k}^2$ intersect with points marked $y_b$ in $I_{h,k}^1$. 
By our technique of marking points with both $x_b,y_b$ as $\neg b$, 
this could give rise to inconsistency. For example, consider $[l,u)=[6,7)$, 
$\tau_j=3.1, \tau_k=4.8, \tau_h=5.9$. Clearly,  $\tau_k-\tau_j, \tau_h-\tau_k \in (1,7]$. 
For $d_1=\lceil \tau_k-\tau_j \rceil+(l-u)=1$, 
the over approximations of the interval $[\tau_j+u, \tau_k+l)=[10.1,10.8)$ are 
$[\tau_k+l-d_1,\tau_k+l)=[9.8,10.8)=I_{j,k}^1$, and  $[\tau_j+u, \tau_j+u+d_1)= 
[10.1,11.1)=I_{j,k}^2$. By construction, points in $[9.8,10.8)=I_{j,k}^1$
are marked $y_b$, points in 
$[10.1,11.1)=I_{j,k}^2$ are marked $x_b$. Clearly, 
points in $[10.1,10.8)$ have both $x_b,y_b$ marked. Again, 
the over approximations for the interval $[\tau_k+u, \tau_h+l)=[11.8,11.9)$ are 
$I^1_{k,h}=[\tau_h+l-d_2,\tau_h+l)=[10.9,11.9)$ and 
$I^2_{k,h}=[\tau_k+u, \tau_k+u+d_2)=[11.8,12.8)$ for $d_2=\lceil \tau_h-\tau_k \rceil+(l-u)=1$.
As per the marking scheme, we would mark 
$[10.9,11.9)$ with $y_b$ and $[11.8,12.8)$ with $x_b$. While this gives us points in 
$[11.8,11.9)$ marked with both $x_b,y_b$, this also gives us points in 
$[10.9,11.1)$ marked with both $x_b,y_b$. We would then mark $\neg b$ for all points in 
$[10.9,11.1)$, giving rise to inconsistency, as 
$[10.9,11.1)$ is marked $b$ by $\mathsf{MARK}_b$.
  However, had we marked 
$[9.8,10.8)=I_{j,k}^1$ with $y_{0b}$, 
$[10.1,11.1)=I_{j,k}^2$ with $x_{0b}$,
 $[10.9,11.9)=I^1_{k,h}$ with $y_{1b}$ and 
 $[11.8,12.8)=I^2_{k,h}$ with $x_{1b}$, 
 the erroneous interval 
$[10.9,11.1)$ is marked with $x_{1b}$ and $y_{0b}$. Thus, using two bits 0,1, we can rule out marking points having $x_{cb},y_{1-c~b}$ 
with $\neg b$.
 The situation of erroneous marking is illustrated in Figure \ref{fig:err}.  

\item Note that it suffices to have only two bits 0,1 and hence propositions 
$x_{0b},y_{0b},x_{1b},y_{1b}$. We do not need  $x_{2b},y_{2b}$. 
Consider any two pairs of points $j,k$ and $h,m$ such that  
$j<k<h<m$, and 
$j$ and $k, h$ and $m$ and $k$ and $h$ are all consecutive with respect to  $a$. i.e, 
there are no points between $j,k$ or $k,h$ or $h,m$ that are marked $a$. Let $\tau_k-\tau_j, \tau_m-\tau_h > u-l$.  
  Assume further that $a_0 \in \sigma_j$ as per our marking scheme.
  % Let $d_1=\lceil \tau_k-\tau_j\rceil+(l-u)$,
%and $d_2=\lceil \tau_m-\tau_h\rceil+(l-u)$. 
%be the closest integers greater than or equal to the actual durations 
%of the intervals $[\tau_j+u, \tau_k+l)$ and  
%$[\tau_h+u, \tau_m+l)$ respectively. 
There are 2 cases :\\
%\begin{itemize}
 \underline{Case 1:} $\tau_h - \tau_k \leq u-l$. In this case $k,h$ will be marked as $a_1$ and $j, m$ will be marked as $a_0$. Note that the regions $I^1_{k,h}$ and $I^2_{k,h}$ are empty. 
No  erroneous intersection can happen :  $I^2_{j,k}$ is marked $x_{0b}$ while $I^1_{h,m}$ is marked $y_{1b}$.\\ 
%($m$ is marked as $a_0$). Hence in this case erroneous intersection will not happen.
\underline{Case 2:} $\tau_h - \tau_k > u-l$. In this case $j,h$ will be marked as $a_0$ and $k,m$ will be marked as $a_1$. Let $d_1= \lceil \tau_k - \tau_j \rceil +(l-u)$,  
$d_2= \lceil \tau_h - \tau_k \rceil +(l-u)$, 
and $d_3=\lceil \tau_m - \tau_h \rceil +(l-u)$.
\begin{itemize}
 \item Intervals $I^1_{h,m}=[\tau_m+l-d_3, \tau_m+l)$ (marked $y_{0b}$) and $I^2_{j,k}=[\tau_j+u, \tau_j+u+d_1)$ (marked $x_{0b}$) are disjoint: we have 
  $\tau_j+u+d_1 < \tau_k+u < \tau_h+l < \tau_m+l-d_3$. 
 \item Intervals $I^2_{h,m}=[\tau_h+u, \tau_h+u+d_3)$ (marked $x_{0b}$) and $I^1_{j,k}=[\tau_k+l-d_1, \tau_k+l)$ (marked $y_{0b}$) are disjoint:  
 $\tau_h+u \geq \tau_k+(u-l)+l \geq \tau_k+l$.
  \end{itemize}
  This shows that for consecutive pairs of $a$ points $j,k$ and $h,m$ 
  where $\tau_k-\tau_j, \tau_h-\tau_k, \tau_m-\tau_h > u-l$, 
  intervals $I^1_{h,m}$ and $I^2_{j,k}$ (respectively $I^2_{h,m}$ and $I^1_{j,k}$) 
  which are marked $x_{ib},y_{ib}$ will never intersect.  
 %\end{itemize}
 \item  The formulae $\mathsf{MARK}_{x,y,c}$ only say where $x_{cb},y_{cb}$ are marked; they do not disallow occurrences of 
 $x_{1-c~b}, y_{1-c~b}$ at those points. 
 We claim that the free occurrences of $x_{1-c~b}, y_{1-c~b}$
 does not create problems. Note that points marked $b$ 
 by $\mathsf{MARK}_b$ and points marked $\neg b$ by 
 $\mathsf{MARK}_{\neg b,c}$, $c \in \{0,1\}$ are disjoint and 
 span $dom(\rho)$. 
Let $p,q$ be consecutive points marked $a$.  
 For every point $p$ with $a \in \sigma_p$, 
 $[\tau_p+l, \tau_p+u)$ is marked $b$ by 
 $\mathsf{MARK}_b$, and $[\tau_p+u, \tau_q+l)$ is marked $\neg b$
 by $\mathsf{MARK}_{\neg b,c}$. In case $p=last_a$, 
 then $[\tau_{last_a}+u, \infty)$ is marked $\neg b$ by $\mathsf{MARK}_{last}$. 
 Thus, inducting on the $a$'s in $\rho$, 
 the union of the points marked $b$ by $\mathsf{MARK}_b$ (call it $B$) and points marked $\neg b$ by 
 $\mathsf{MARK}_{\neg b,c}$ (call it $\bar{B}$) is $dom(\rho)$. 
  
 Thus, there are 2 possibilties for 
 the free occurrence of $x_{1-c~b}, y_{1-c~b}$: 
 \begin{itemize}
  \item $x_{1-c~b}, y_{1-c~b}$ occur freely 
 in $\bar{B}$. The freely occurring $x_{1-c~b}, y_{1-c~b}$ 
results in marking of $\neg b$ by formula $\mathsf{MARK}_{\neg b,~1-c}$; 
 this does not generate inconsistencies, since 
 they are already marked $\neg b$ by $\mathsf{MARK}_{\neg b,c}$. 
 \item   $x_{1-c~b}, y_{1-c~b}$ occur freely 
 in $B$. The freely occurring $x_{1-c~b}, y_{1-c~b}$ 
results in marking of $\neg b$ by formula $\mathsf{MARK}_{\neg b,~1-c}$;
but these points are already marked $b$ by $\mathsf{MARK}_b$. Thus,
at any point $p$ in $B$, 
$\rho,p \nvDash x_{1-c~b} \wedge y_{1-c~b}$, for $c \in \{0,1\}$.
Thus, at all points of $B$, the appearance of $x_{1-c~b}$ and  $y_{1-c~b}$ (if that is the case),  
is mutually exclusive. 
    \end{itemize}
 Thus, free markings of $x_{cb},y_{cb}$ if at all, they appear, do not come in the way 
 of correctly marking points with $b, \neg b$.
  \end{enumerate}
 The formula $\mathsf{MARK}$ in $\mathsf{MTL}[\until_I,\since]$ obtained as a conjunction of 
  $\mathsf{MARK}_b, \mathsf{MARK}_{first}, \mathsf{MARK}_{last},
 \mathsf{MARK}_a,\mathsf{MARK}_{x,y,c}$,
 $\mathsf{MARK}_{beg,end,d},
 \mathsf{MARK}_{succ,\infty},\mathsf{MARK}_{\neg b,c}$ and 
 $\mathsf{MARK}_{\neg b, \infty}$.
   is such that $\rho \models ENF_{\Sigma}(T)$ iff $\rho' \models \mathsf{MARK}$, where $\rho' \setminus X=\rho$.
   A proof of correctness can be found in Appendix \ref{correct}. 
 Using the technique in \cite{deepak08}, we can eliminate the $\since$ modality from 
  $\mathsf{MARK}$ obtaining an equisatisfiable formula $\psi$ in $\mtl$. 
     \end{proof}
%Note that the above formulae also hold for infinite and weakly monotonic words.
 Note that our   reduction  does not  introduce any new punctual modality. 
Hence, we also have  the equivalence modulo simple projection of $\mathsf{MITL}[\until_I,\since_I]$ and 
 $\mathsf{MITL}[\until_I]$.  
\begin{theorem}
\label{theo:main}
For every $\varphi \in \mtlsns$ over $\Sigma$, we can construct $\psi_{fut}$ in $\mtl$ 
over $\Delta = \Sigma \cup \Sigma'$ such that $\varphi = \exists (\Sigma'-\Sigma). \psi_{fut}$.
\end{theorem}
\begin{proof}
Follows from the fact that  $\since_{np}$ can be expressed using $\since$ and 
$\past_{np}$ \cite{deepak08}
%. For instance, we can write $a \since_{[l, r)}b$ as 
%$\past_{[l, r)}b ~\wedge ~(a \since b) \land \boxminus_{[0,l)}(a \wedge a \since b)$, for $r=l+1, \infty$. 
%Similarly,  all intervals $\langle l, l+1 \rangle$, $\langle l, \infty \rangle$
%are handled.  Further,
and elimination of $\since$ modulo simple projections  \cite{deepak08}, \cite{formats11}.
%\ref{since-rem}
 %Hence,
%the only past modalities in the formulae are $\past_{[l,\infty)}$ or $\past_{[l,u)}$.
%Lemmas \ref{remove-pastinf2} and \ref{past-b2} show how these can be expressed
%in $\mtl$ to obtain formula $\psi$ equisatisfiable to $\varphi$. 
\end{proof}

\subsection{Simple Versus Oversampling Projections: Formulae Size}
\label{compl}
Consider a formula $\varphi \in \mathsf{MTL}[\until_I, \past_{np}]$. 
Assume that the number of past modalities in $\varphi$ is $n$, of which there are  
$n_b$  bounded past modalities and $n_u$ unbounded past modalitties. 
i.e, $n=n_b+n_u$. 
  Flattening $\varphi$ 
  results in a linear increase in the size of $\varphi$. 
  Converting $\varphi_{flat}$ to  $ENF$ gives a constant 
size increase. 
Elimination of unbounded past  
(Lemma \ref{remove-pastinf2})  
also results in a constant increase in size.  
 During elimination of bounded past modalities $\past_{[l,u)}$ (Lemma \ref{past-b2}), 
  we add $l-1$ new formulae resulting in $\mathcal{O}(l)$ extra modalities. 
   Thus, the number of extra modalities introduced after elimination 
   of all the $n_b$ temporal definitions corresponding to bounded past modalities is $\leq n l_{max}$, where 
   $l_{max}$ is the maximal lower bound of all bounded past modalities in $\varphi$.
   Hence, the formula obtained by simple projections, $\psi$ has in the worst case, an exponential increase in  size over $\varphi$. 
 In the case of oversampled projections,   
   it can be seen that both bounded as well as unbounded past modalities 
   contribute to a linear increase in the size of the resultant formulae. 
 In simple projections (Lemma \ref{past-b2}), marking $\neg b$ correctly in $[\tau_j+u, \tau_k+l)$  
 depended on the distance $\tau_k-\tau_j$, resulting in $l-1$ formulae; 
 in the case of oversampling projections (Lemma \ref{past-b1}), this is handled indirectly
   by the introduction of extra integral points  between $j$ and $k$. However, 
the formulae needed to introduce these extra points correctly 
have a constant size.  A more detailed complexity analysis can be found in  Appendix \ref{compl}.

\subsection{Expressiveness}
We wind up this section with a brief discussion 
about the expressive powers of logics 
$\mtlsns$ and $\mtluns$. The following lemma highlights 
that even unary modalities $\fut_I ,\past_I$ with singular intervals 
are more expressive than $\until_{np}, \since_{np}$; likewise, 
non-singular intervals are more expressive than intervals 
of the form $[0, \infty)$.

\begin{lemma}
\label{game:proof}
(i) $\mtlfutpw \nsubseteq \mtluns$, 
 (ii) $\mtlfutp \nsubseteq \mtlsns$, and (iii) $\mitlfp \nsubseteq \mathsf{MTL}[\until_I,\since]$.
\end{lemma}
\begin{proof}
The formula $\fut_{(0,1)}\{a \wedge \neg \fut_{[1,1]}(a \vee b)\}$
in $\mtlfutpw$ has no equivalent formula in $\mtluns$. 
Similarly, the formula  $\fut\{b \wedge  \neg \past_{[1,1]}(a \vee b)\}$
in $\mtlfutp$ has no equivalent formula in $\mtlsns$. 
The formula $\fut_{(1,2)}[a \wedge \neg \past_{(1,2)}a] \in \mitlfp$ 
has no equivalent formula in $\mathsf{MTL}[\until_I,\since]$. Details in Appendix \ref{games}. 
\end{proof}

\section{Discussion}
In this paper, we have proposed two notions of equivalence between $\mathsf{MTL}$ formulae (with different sets of propositions), which both preserve satisfiability. The notion $\phi = \exists X. \psi$, denoting {\em equisatisfiability modulo simple projection} denotes that a timed word satisfying $\phi$ can be extended to a timed word with additional propositions $X$ which satisfies $\psi$, and a timed word satisfying $\psi$ can be projected to a timed word satisfying $\phi$. In both cases the set of time stamps of the letters remains identical.
A more elaborate notion, $\phi = \exists \downarrow X . \psi$, denoting {\em equisatisfiability modulo oversampling projection}, is similar but the models of $\psi$ may have additional time points. Thus, during temporal projection we allow oversampling of the original behaviour by adding new time points. Both forms of temporal projections are useful. They often allow formulae of a more complex logic to be effectively reduced in equi-satisfiable manner to formulae of a much simpler logic. This often provides a convenient technique for proving satisfiability.
As a significant use of this technique of temporal projections, in the paper, we have shown the decidability of $\mtlsns$ over finite strictly monotonic timed words. This logic is more expressive than the previously known decidable fragments of $\mathsf{MTL}$ as well as $\mathsf{MITL}$ but less expressive than $\mtlfull$. A symmetric proof would allow showing that $\mtluns$ is also decidable. Our result can also be adapted to weakly monotonic finite timed words (see Appendix \ref{weak}). Thus, we have extended the boundary of known decidable fragments of logic $\mathsf{MTL}$ over timed words. We note that the proof techniques used for showing decidability of $\mathsf{MTL}$ as well as $\mathsf{MITL}$, 
do not seem to generalize easily to the logic $\mtlsns$ considered here.
In proving decidablity of $\mtlsns$, we have given two different proofs. In the first proof, we reduced $\mtlsns$ to $\mathsf{MTL}[\until_I]$ using the notion of oversampled temporal projections. This encoding is relatively simple and results only in linear blowup in formula size. We also gave an alternative reduction using only simple temporal projections, but the reduction turns out to be considerably more complex, and leads to an exponential blow up in formula size.

The technique of temporal projections has been widely used for continuous time $\mathsf{MTL}$. For example, Hirshfeld and Rabinovich \cite{fundinfo04} used it to eliminate non-singular future operator $\fut_{[0,1)}$ in terms of $\past_{[0,1)}$, $\until$ and $\since$. Subsequently, D'souza {\em et al} \cite{deepak08} as well as Kini {\em et al} \cite{formats11} used the technique to remove past operator $\since_I$ from $\gmtlfull$. Their reduction does not carry over to logic $\mtlsns$ over pointwise time which is expressively weak and allows insertion errors.
In this paper, we have extended the technique of temporal projections to pointwise time (timed words). One novel aspect of our formulation is that during temporal projection we allow oversampling of the original behaviour by adding new time points. We have demonstrated that the ability of adding such additional points can considerably simplify the reductions. The expressive power of (the two forms of) temporal projections is an interesting topic of future work.

\newpage
\appendix

\subsection{Proof of Lemma \ref{lem:boolclosedequis}}
\label{proofs-1}
We first define the composition of two simple extensions. 

\noindent{\it \underline{Composition of Two Simple Extensions}}:
Consider $\Sigma,X_1,X_2$ such that $\Sigma\cap X_1 =\emptyset$ and $\Sigma \cap X_2 =\emptyset$ . Let $\Sigma_1 = \Sigma \cup X_1$, $\Sigma_2 = \Sigma \cup X_2$
and $X=X_1\cup X_2$. 

Let $\zeta'_1=(\sigma^1,\tau^1)$ be a $(\Sigma, X_1)$-simple extension  and 
let $\zeta'_2=(\sigma^2,\tau^2)$  be a $(\Sigma,X_2)$-simple extension, such that \\ 
$\zeta'_1 \setminus X_1=\zeta'_2 \setminus X_2$. 
Then the composition of $\zeta'_1$ and $\zeta'_2$ 
denoted $\zeta'_1\oplus \zeta'_2$,  
is a  $(\Sigma,X)$-simple extension $\zeta'=(\sigma', \tau')$ 
such that $\sigma'_i=\sigma^1_i \cup \sigma^2_i$ and 
$\tau_i=\tau^1_i=\tau^2_i$.
Note that $(\zeta'_1 \oplus \zeta'_2) \setminus X=
\zeta'_1 \setminus X_1=\zeta'_2 \setminus X_2$, and 
 $\zeta' \setminus X_2 = \zeta'_1$ and $\zeta' \setminus X_1 = \zeta'_2$. 
Consider the following example:
\begin{itemize}
\item 
Let $\Sigma=\{a,b\}, X_1=\{c\}, X_2=\{d\}$,  
\item $\zeta'_1=(\{a\},0.3)(\{b,c\},0.8)(\{b\},1.1)$, and    
\item $\zeta'_2= (\{a,d\},0.3)(\{b,d\},0.8)(\{b\},1.1)$. Then 
\item 
$\zeta'_1 \setminus X_1=\zeta'_2 \setminus X_2=
(\{a\},0.3)(\{b\},0.8)(\{b\},1.1)$,
\item   
$\zeta'_1 \oplus \zeta'_2=(\{a,d\},0.3)(\{b,c,d\},0.8)(\{b\},1.1)$,
\item $(\zeta'_1\oplus \zeta'_2) \setminus X_2=
(\{a\},0.3)(\{b,c\},0.8)(\{b\},1.1)=\zeta'_1$,
\item  
$(\zeta'_1\oplus \zeta'_2) \setminus X_1=
(\{a,d\},0.3)(\{b,d\},0.8)(\{b\},1.1)=\zeta'_2$. 
\end{itemize}
We use the following easy lemma in the proof:
\begin{lemma}
\label{obs-simple}
Consider $\Sigma,X_1,X_2$ such that $\Sigma\cap X_1 =\emptyset$ and $\Sigma \cap X_2 =\emptyset$ . Let $\Sigma_1 = \Sigma \cup X_1$, $\Sigma_2 = \Sigma \cup X_2$
and $X=X_1\cup X_2$. 
Then $\zeta' = (\zeta' \setminus X_2) \oplus (\zeta' \setminus X_1)$. 
\end{lemma}
\noindent \underline{Proof of Lemma \ref{lem:boolclosedequis}}:
\begin{proof}
Assume $\varphi_1 = \exists X_1. \psi_1$, $\varphi_2 = \exists  X_2. \psi_2$. 
\begin{itemize}
\item[(a)]
Then, for $i=1,2$, 
and any timed word $\rho_i$ over $\Sigma$, 
such that $\rho_i \models \varphi_i$, 
we have a 
$(\Sigma, X_i)$-simple extension $\rho'_i$ such that $\rho'_i \models \psi_i$ and 
$\rho'_i \setminus X_i=\rho_i$.
\item[(b)]  For 
any timed word $\rho'_i$, $\rho'_i \models \psi_i$ implies 
$\rho'_i$ is a $(\Sigma,X_i)$-simple extension  
such that $\rho'_i \setminus X_i \models \varphi_i$.
\end{itemize}
Consider $\varphi_1 \wedge \varphi_2$, a formula built over $\Sigma$.  Also, $\psi_1 \wedge \psi_2$ 
is a formula built from  $\Sigma \cup X_1 \cup X_2$. 
Let $X=X_1 \cup X_2$.
\begin{enumerate}
\item Let $\zeta'$ be a timed word over $\Sigma \cup X$  
such that $\zeta' \models \psi_1 \wedge \psi_2$. Then 
$\zeta' \models \psi_i$ for $i=1,2$.  
Since $\psi_i$ is a formula built from $\Sigma_i$, and 
$X_1 \cap X_2=\emptyset$, we have 
\begin{itemize}
\item $\zeta' \models \psi_1 \rightarrow \zeta'\setminus X_2 \models \psi_1$, and 
\item $\zeta' \models \psi_2 \rightarrow \zeta'\setminus X_1 \models \psi_1$.
\item Call $\zeta' \setminus X_1=\zeta'_2$ and $\zeta' \setminus X_2=\zeta'_1$.
\end{itemize}
Note that $\zeta'_1$ is a $(\Sigma,X_1)$-simple extension 
and $\zeta'_2$ is a $(\Sigma,X_2)$-simple extension. 
This gives,  by (b) above that $\zeta'_1 \setminus X_1 \models \varphi_1$ 
and $\zeta'_2 \setminus X_2 \models \varphi_2$.
 By Lemma \ref{obs-simple}, we have $\zeta'_1 \setminus X_1=\zeta'_2\setminus X_2$, call 
 it some timed word $\zeta$ over $\Sigma$. Then $\zeta \models \varphi_1 \wedge \varphi_2$.
 Also, $\zeta=(\zeta'_1\setminus X_1)=(\zeta' \setminus X_2)\setminus X_1=\zeta' \setminus X$.
 
\item  Now let $\zeta$ be a timed word over $\Sigma$ such that $\zeta \models \varphi_1 \wedge \varphi_2$. 
 We have to show that there is a $(\Sigma, X)$-simple extension $\zeta'$ 
 such that $\zeta' \models \psi_1 \wedge \psi_2$ such that $\zeta' \setminus X=\zeta$. 
  
  Since   $\varphi_1 = \exists X_1. \psi_1$, $\varphi_2 = \exists  X_2. \psi_2$, we know that 
 for any word $\zeta$ over $\Sigma$ satisfying $\varphi_1 \wedge \varphi_2$,
 $\zeta \models \varphi_i$.
 By (a) above, $\zeta \models \varphi_i$ 
   implies there exists  $(\Sigma,X_i)$-simple  extensions  
  $\zeta'_i$ such that $\zeta'_i \models \psi_i$, with $\zeta'_i \setminus X_i=\zeta$.
     Then the composition $\zeta'=\zeta'_1 \oplus \zeta'_2$, of  
 $\zeta'_1$ and  $\zeta'_2$ is well-defined.
  Clearly, $\zeta'$ is a $(\Sigma, X)$-simple extension obtained by composing
 the $(\Sigma, X_1)$-simple extension $\zeta'_1$ and   
 the $(\Sigma, X_2)$-simple extension $\zeta'_2$ 
  such that   $\zeta' \setminus X = \zeta$. 
 
 Since $X_1\cap X_2=\emptyset$, 
 and $\psi_1$ is built from from $\Sigma \cup X_1$ and 
 $\psi_2$ from $\Sigma \cup X_2$,  $\zeta'_1 \oplus \zeta'_2$   
 will not interfere in the satisfiability of either $\psi_1$ or $\psi_2$, in a way different from 
 $\zeta'_1$ and $\zeta'_2$:  
 Assume the contrary. 
 That is, $\zeta' \nvDash \psi_1 \wedge \psi_2$. That is, $\zeta' \nvDash \psi_1$ 
 or $\zeta' \nvDash \psi_2$. Let $\zeta' \nvDash \psi_1$. 
 If so, then $\zeta' \setminus X_2 \nvDash \psi_1$ since $\psi_1$ has no symbols 
 from $X_2$ (by assumption $X_1$ and $X_2$ are disjoint). 
 But $\zeta' \setminus X_2= \zeta'_1$, and 
 we know $\zeta'_1 \models \psi_1$, contradicting $\zeta' \nvDash \psi_1$. Hence, 
 $\zeta' \models \psi_1 \wedge \psi_2$. 
  \end{enumerate}    
   The following example illustrates what might go wrong when 
   $X_1 \cap X_2 \neq \emptyset$. 
        Consider,   
   $\Sigma=\{a,c,d\}$, $X_1=\{b,e\}$ and $X_2=\{b,f\}$. 
   Note that $X_1 \cap X_2=\{b\}$. 
   
     Consider formulae 
   $\psi_1=b \wedge \Box (b \leftrightarrow \fut c) \wedge \Box \bigvee \Sigma$ and 
   $\psi_2= b \wedge \Box (b \leftrightarrow \fut a) \wedge \Box \bigvee \Sigma$. 
   Also, let $\varphi_1=\fut c$ and $\varphi_2=\fut a$. 
   Let $\zeta$ be the word $(d,0.1)(c,0.3)(d,0.7)(a,0.9)$ over $\Sigma$. 
  Clearly, $\zeta \models \varphi_1 \wedge \varphi_2$.
   
    Consider $\zeta'_1=(\{d,e,b\},0.1)(\{c\},0.3)(\{e,d\},0.7)(\{a\},0.9)$, a $(\Sigma, X_1)$-simple extension 
   and the 
     $(\Sigma, X_2)$-simple extension
      $\zeta'_2=(\{d,f,b\},0.1)(\{c,b\},0.3)(\{f,b,d\},0.7)(\{a\},0.9)$. 
 Then,  $\zeta'_1 \models  \psi_1$, $\zeta'_2 \models \psi_2$,  
 $\zeta'_1 \setminus X_1=\zeta'_2 \setminus X_2=\zeta$. 
 However,  \\ $\zeta'=    
    (\{d,b,e,f\},0.1)(\{c,b\},0.3)(\{d,b,e,f\},0.7)(\{a\},0.9)$, \\ 
    the composition of $\zeta'_1$ and $\zeta'_2$ is such that  $\zeta' \nvDash (\psi_1 \wedge \psi_2)$. 
       \end{proof}

 \subsection{Proof of Lemma \ref{lem:gen}}
 \label{proofs-2}
 \begin{proof}
 The proof follows by structural induction on $\varphi$. 
\begin{itemize} 
\item Let $\rho$ be a timed word over $\Sigma$ such that $\rho \models \varphi$. 
We have to show that for all $(\Sigma,X)$-oversampled behaviour $\rho'$ 
such that $\rho' \downarrow X=\rho$ holds,  $\rho' \models ONF_{\Sigma}(\varphi)$.   

Consider a  $(\Sigma,X)$-oversampled behaviour $\rho'$, 
such that $\rho' \downarrow X=\rho$. Then, there exists a strictly increasing function 
$f:\{1,2,\dots,n\} \rightarrow \{1,2,\dots,m\}$ such that $n=|dom(\rho)|$, $m=|dom(\rho')|$, and  
\begin{itemize}
\item $f(1)=1$, $\sigma_1=\sigma'_1 \cap \Sigma$, $\tau_1=\tau'_1$, and 
\item $f(n)=m$, $\sigma_n=\sigma'_m \cap \Sigma$, $\tau_n=\tau'_m$, and 
 \item For $1 \leq i \leq n-1$, $f(i)=j$ and $f(i+1)=k$ iff 
 \begin{itemize}
 \item $\sigma_i=\sigma'_j \cap \Sigma$,  
 %or $\sigma_i=\emptyset$, $\emptyset \neq \sigma'_j \subseteq Y$,
  and $\tau_i=\tau'_j$,
 \item $\sigma_{i+1}=\sigma'_k \cap \Sigma$, 
 %or $\sigma_{i+1}=\emptyset$, 
 %$\emptyset \neq \sigma'_k \subseteq Y$ 
 and   $\tau_{i+1}=\tau'_k$, 
 \item For all $j < l < k$, $\sigma'_l \subseteq X$. 
 \end{itemize}
\end{itemize}
  
By applying structural induction on depth of $\varphi$, we  show that 
$\rho \models \varphi \rightarrow \rho' \models ONF_{\Sigma}(\varphi)$.
 For depth 0, the base case trivially holds for atomic propositions. For example if $\varphi=a \in \Sigma$, then 
 $ONF_{\Sigma}(\varphi)=a \wedge act$. Clearly, $\rho, 1 \models a$ iff $\rho', f(1) \models ONF_{\Sigma}(a)$. 
 
 Assume the result for formulae $\varphi$ of depth $\leq n-1$. Consider $\varphi$ as a formula of depth $n$. 
  Lets  discuss the case  
of formulae of the form $\varphi=\psi_1 \until_I \psi_2$ where $\psi_1$ and $\psi_2$ have depth $\leq n-1$.
        
  If $\rho,i \models \psi_1 \until_I \psi_2$, 
then there exists $j > i$ where $\psi_2$ holds, and all points in between $i$ and $j$ satisfy $\psi_1$. Also,
  $t_j - t_i \in I$. By the above, there exists a point $f(j) > f(i)$ such that $\sigma'_{f(j)} \models ONF_{\Sigma}(\psi_2)$ (by induction hypothesis), and $\sigma'_{f(j)} \models act$ (definition of $f$).   
Let $\{i_1, \dots, i_q\}$ be the set of points between $f(i)$ and $f(j)$. 
For all $i < l < j$,  $f(l) \in \{i_1, \dots, i_q\}$.  Also, $\sigma_{f(l)}\models ONF_{\Sigma}(\psi_1)$. However, there are points 
$i_j \in   \{i_1, \dots, i_q\}$ such that $i_j \neq f(l)$ for any $i < l < j$. These points 
are such that $\sigma'_{i_j} \cap \Sigma =\emptyset$. Now if we look at points between 
$f(i)$ and $f(j)$, then we have 
\begin{itemize}
\item For all points $k$ such that $f(i) < k < f(j)$, we have $\sigma'_k \models ONF_{\Sigma}(\psi_1)$,  
or $\sigma'_k \cap \Sigma = \emptyset$.  \\
i.e, $(\sigma'_k \cap \Sigma \neq \emptyset) \rightarrow 
 \sigma'_k \models ONF_{\Sigma}(\psi_1)$.  
 \item Recall that if $\sigma'_k \cap \Sigma \neq \emptyset$, then $\sigma'_k \models act$
\end{itemize}
  The above conditions give us \\
  $\rho', f(i) \models (act \rightarrow ONF_{\Sigma}(\psi_1))\until_I (act \wedge ONF_{\Sigma}(\psi_2))$. 
Also, since $\rho'$ is a $(\Sigma,X)$-oversampled behaviour, $act$ holds good at the start and end points. 
$\rho' \models act$ iff $act$ holds good at the starting point. $\Box \bot$ holds good only at the last point; $\bot$ stands for $false$. Clearly, $\rho \models (\psi_1 \until_I \psi_2)$ implies 
$\rho' \models     (act \rightarrow ONF_{\Sigma}(\psi_1)\until_I (act \wedge ONF_{\Sigma}(\psi_2) \wedge act \wedge (\Box \bot \rightarrow act)$.
%The proof extends straightforwardly to formulae of the form $\varphi_1 \until_I \varphi_2$. 
 The proof for past modaility is analogous.
   
 \item Let $\rho'$ be a $(\Sigma,X)$-oversampled behaviour 
 such that $\rho' \models ONF_{\Sigma}(\varphi)$. 
 We have to show that $\rho' \downarrow X \models \varphi$.   
    In a manner similar to the above, by structural induction of $\varphi$, we 
    can show that $\rho' \downarrow X \models \varphi$. 
     \end{itemize}     
     \end{proof}      

 \subsection{Proof of Lemma \ref{onf1}}
\label{proofs:onf1}
\begin{proof}
Follows from Lemma \ref{lem:gen} and  equivalence 
 of $\zeta$ and $ONF_{\Sigma}(\zeta)$. 
 \end{proof}
\subsection{Proof of Lemma \ref {lem:boolclosedequis-2}}
\label{proofs-3}
We first define the composition of two oversampled behaviours. 

\noindent{\it \underline{Composition of Oversampled Behaviours}}: 
Let $\rho_1=(\sigma^1,\tau^1)$ be a $(\Sigma,X_1)$-oversampled behaviour and  
$\rho_2=(\sigma^2,\tau^2)$ be a $(\Sigma,X_2)$-oversampled behaviour such that 
$\rho_1 \downarrow X_1=\rho_2 \downarrow X_2$.  
This condition says that the points in $\rho_1$ where propositions of $\Sigma$ hold 
is exactly same as the points in $\rho_2$ where propositions of $\Sigma$ hold; moreover 
the same propositions of $\Sigma$ hold at these points of $\rho_1$ and $\rho_2$. 
Let $\Sigma_1=\Sigma \cup X_1$ and $\Sigma_2=\Sigma \cup X_2$.
We define the composition of $\rho_1$ and $\rho_2$ denoted 
$\rho_1 \boxplus \rho_2$ to be 
all $(\Sigma,X_1\cup X_2)$-oversampled behaviours 
$\rho$ such that 
$\rho \downarrow X_1 = \rho_2$ and $\rho\downarrow X_2 = \rho_1$.
  Note that $\rho_1 \boxplus \rho_2$ is guaranteed to exist only when $X_1 \cap X_2=\emptyset$. 
  The following example illustrates that when $X_1 \cap X_2\neq \emptyset$, $\rho_1 \boxplus \rho_2$ 
 may not exist.    
  
Consider $\Sigma=\{a,b\}, X_1=\{c,e\},X_2=\{d,e\}$. 
Let $\rho_1=(\{a,c\},0.1)(\{e\},0.3)(\{b,e,c\},1)$ be a $(\Sigma,X_1)$-oversampled behaviour and 
 $\rho_2=(\{a\},0.1)(\{e\},0.3)(\{b,e,d\},1)$ be a $(\Sigma,X_2)$-oversampled behaviour. 
 Then $\rho_1 \downarrow X_1=
(\{a\},0.1)(\{b\},1)=\rho_2 \downarrow X_2$. 
Assume that $\rho \in \rho_1 \boxplus \rho_2$. Then, 
$\rho \downarrow X_1= \rho_2$. 
However, $\rho \downarrow X_1$ will not contain 
any position $i$ which is marked just with $e$, since 
such a position will be eliminated during oversampled projection with respect to $X_1$. 
%Hence, restricting $\rho \downarrow X_1$ to $\Sigma_2$ also will not contain such points; 
Thus, there can be no such $\rho$, which 
after oversampling projections with respect to $X_1$ will give $\rho_2$.
A similar problem happens 
when trying to show that $\rho \downarrow X_2=\rho_1$.

We now give an example to illustrate the composition 
of two oversampled behaviours. Let $\Sigma{=}\{a\}, X_1{=}\{c\}, X_2{=}\{d\}$,
$\rho_1{=}(\{a\},0.1)(\{c\},0.5)$ and $\rho_2{=}(\{a\},0.1)(\{d\},0.5)(\{d\},0.5)$. 
$\rho_1 \boxplus \rho_2$ consists of:
\begin{itemize}
\item $(\{a\},0.1)(\{c\},0.5)(\{d\},0.5)(\{d\},0.5)$
\item $(\{a\},0.1)(\{d\},0.5)(\{d\},0.5)(\{c\},0.5)$
\item $(\{a\},0.1)(\{d\},0.5)(\{c\},0.5)(\{d\},0.5)$
\end{itemize}
Clearly, when the words $\rho_1, \rho_2$ are weakly monotonic, 
$\rho_1 \boxplus \rho_2$ can consist of more than one word;
however, when $\rho_1, \rho_2$  are strictly monotonic,
$\rho_1 \boxplus \rho_2$ is a unique word. Our proof applies to both 
weakly monotonic and strictly monotonic words.
We use the following easy lemma in the proof:
\begin{lemma}
\label{obs}
Let $X_1 \cap X_2=\emptyset$, and $X_1 \cup X_2=X$. 
Let $\rho$ be a $(\Sigma,X)$-oversampled behaviour, and 
let $\Sigma_1=\Sigma \cup X_1$ and $\Sigma_2=\Sigma \cup X_2$.  
Then $\rho \in [(\rho \downarrow X_2)] \boxplus  
[(\rho \downarrow X_1)]$.
\end{lemma}

\noindent \underline{Proof of Lemma \ref {lem:boolclosedequis-2}}:
\begin{proof}

Given $\varphi_1=\exists \downarrow X_1. \zeta_1$ and  
  $\varphi_2 = \exists \downarrow X_2. \zeta_2$. We know that 
  by definition, 
  \begin{itemize}
  \item[(a)] For any $(\Sigma,X_i)$-oversampled behaviour $\rho'_i$, \\
  $\rho'_i \models \zeta_i \rightarrow (\rho'_i \downarrow X_i)\models \varphi_i$.
  \item[(b)] For any timed word $\rho_i$ over $\Sigma$ such that $\rho_i \models \varphi_i$, there exists 
  a $(\Sigma,X_i)$-oversampled behaviour $\rho'_i$ such that $\rho'_i \models \zeta_i$  and $\rho'_i   \downarrow X_i=\rho_i$.
    \end{itemize}
  
  We now want to show that 
  $\varphi_1 \wedge \varphi_2 = \exists \downarrow X. 
  (\zeta_1 \wedge \zeta_2)$.
  \begin{enumerate}
 \item Let $\rho$ be a timed word over $\Sigma$ such that $\rho \models \varphi_1 \wedge \varphi_2$. 
 Since $\rho \models \varphi_i$, we have by (b) above, 
  $(\Sigma,X_i)$-oversampled behaviours $\rho'_i$ such that $\rho'_i \models \zeta_i$ and 
  $\rho'_i   \downarrow X_i=\rho$, for $i=1,2$.  
 Hence, $\rho'_1 \boxplus \rho'_2$ is welldefined; 
  let $\rho' \in \rho'_1\boxplus \rho'_2$. 
                 
           Since  $\zeta_i$ is in the oversampled normal form with respect to $\Sigma_i$, 
    by Lemma \ref{onf1},
   we have 
      $\zeta_1=\forall \downarrow. \zeta_1$ and 
    $\zeta_2=\forall \downarrow. \zeta_2$. We already have 
       $\rho'_i \models \zeta_i$, for $i=1,2$. Hence, 
     \begin{itemize}
       \item  any $(\Sigma_1,X_2)$-oversampled behaviour $\rho''$  
     such that $\rho'' \downarrow X_2=\rho'_1$ 
     will also satisfy $\zeta_1$. 
     \item any $(\Sigma_2,X_1)$-oversampled behaviour $\rho'''$  
     such that $\rho''' \downarrow X_1=\rho'_2$ 
     will also satisfy $\zeta_2$.
     \item By definition of $\boxplus$, we know that $\rho' \in \rho'_1\boxplus \rho'_2$ 
     is such that $\rho' \downarrow X_2=\rho'_1$ and 
     $\rho' \downarrow X_1=\rho'_2$.
     \item Picking $\rho'=\rho''=\rho'''$, we have $\rho' \models \zeta_1$ and $\rho' \models \zeta_2$. 
     \end{itemize}
     Hence $\rho' \in \rho'_1 \boxplus \rho'_2$ satisfies $\zeta_1 \wedge \zeta_2$.
    Further, \\ $\rho' \downarrow X= 
  \{[\rho' \downarrow X_1]\downarrow X_2\}=
\{\rho'_2 \downarrow X_2\}=\rho$. 
  
   \item
    Conversely, let $\rho'$ be a $(\Sigma,X)$-oversampled behaviour, such that 
    $\rho' \models \zeta_1 \wedge \zeta_2$. Then $\rho' \models \zeta_i$ for $i=1,2$.
        Again, since  $\zeta_i$ is in the oversampled normal form with respect to $\Sigma_i$, 
    by Lemma \ref{onf1},
   we have 
      $\zeta_1=(\forall \downarrow). \zeta_1$ and 
    $\zeta_2=(\forall \downarrow). \zeta_2$.
     We already have $\rho' \models \zeta_i$ for $i=1,2$. Hence,  
     \begin{itemize}
     \item $\rho' \models \zeta_1 \rightarrow \rho' \downarrow X_2 \models \zeta_1$.
    \item $\rho' \models \zeta_2 \rightarrow \rho' \downarrow X_1 \models \zeta_2$.
     \item  Let $\rho'_1=\rho' \downarrow X_2$ and 
    $\rho'_2=\rho' \downarrow X_1$. Then $\rho'_1 \models \zeta_1$ and 
    $\rho'_2 \models \zeta_2$.
    \item By (a) above, we have   $\rho'_1 \downarrow X_1 \models \varphi_1$ and 
      $\rho'_2 \downarrow X_2 \models \varphi_2$.
   \item By Lemma \ref{obs}, $\rho' \in \rho'_1 \boxplus \rho'_2$. Hence, by definition of $\boxplus$,  
    $\rho'_1 \downarrow X_1=
      \rho'_2 \downarrow X_2$. Call it $\rho$, a timed word over $\Sigma$. Clearly, 
      $\rho \models \varphi_1 \wedge \varphi_2$ and $\rho=\rho' \downarrow X$. 
       \end{itemize}
     
          \end{enumerate}

\end{proof}

\subsection{Proof of Lemma \ref{lemmapast}}
\label{proofs-6}
\begin{proof}
We prove the lemma for intervals of the form $[l,u)$.  The proof can be extended for other type of intervals also.
Assume that $\rho, i \models \past_{[ l, u )}\alpha$. We then show that $\neg(\tau_i < \tau_{first_{\alpha}}+l)$ and 
$\neg(\tau_i \geq  \tau_{last_{\alpha}} +u)$ and $\neg(\tau_i \in [\tau_j+u, \tau_k+l))$ for consecutive points $j,k$ where $\alpha$ holds.
%Part 1 ($\leftarrow$) 
\begin{enumerate}
\item   Let $\tau_i < \tau_{first_{\alpha}}+l$.  
$\rho, i \models \past_{[ l, u )}\alpha$ implies that there is a point $i'$ such that $\tau_{i'} \in (\tau_i-u, \tau_i-l]$, such that $\rho, i' \models \alpha$. Then, $\tau_{i'} \leq \tau_i-l <\tau_{first_{\alpha}}$, contradicting 
that $first_{\alpha}$ is the first point where $\alpha$ holds.   
\item Let $\tau_i \geq \tau_{last_{\alpha}} +u$. Again, $\rho, i \models \past_{[ l, u )}\alpha$ implies that there is a point $i'$ such that 
$\tau_{i'} \in (\tau_i-u, \tau_i-l]$ such that $\rho,i' \models \alpha$. We then have $\tau_{i'} > \tau_i -u \geq \tau_{last_{\alpha}}$, contradicting 
that $last_{\alpha}$ is the last point where $\alpha$ holds. 
\item Assume that there exist consecutive points $j < k$ where $\alpha$ holds. Also, let $\tau_{i} \in [\tau_j+u, \tau_k+l)$.  
$\rho, i \models \past_{[ l, u )}\alpha$ implies that there exists a point $i'$ such that $\tau_{i'} \in (\tau_i-u, \tau_i-l]$ and $\rho,i' \models \alpha$.
Also, $\tau_i-u \in [\tau_j, \tau_k+(l-u))$  and $\tau_i-l \in [\tau_j+(u-l),\tau_k)$. This gives $\tau_j < \tau_{i'} < \tau_k$ contradicting the assumption that 
$j,k$ are consecutive points where $\alpha$ holds. 
\end{enumerate}
Conversely, assume that $\neg(\tau_i < \tau_{first_{\alpha}}+l)$ and $\neg(\tau_i \geq \tau_{last_{\alpha}}+u)$ and $\neg(\tau_i \in [\tau_j+u, \tau_k+l))$ for consecutive points $j,k$ where $\alpha$ holds. Then, 
$\tau_i \in [\tau_{first_{\alpha}}+l, \tau_{last_{\alpha}}+u)$. We show that $\rho,i \models \past_{[l,u)}\alpha$.  
\begin{enumerate}
\item If $\tau_{first_{\alpha}}=\tau_{last_{\alpha}}$, then $\tau_i-u < \tau_{first_{\alpha}} \leq \tau_i -l$. Clearly, 
 $\alpha$ holds in $(\tau_i -u, \tau_i-l]$, and hence $\rho,i \models \past_{[l,u)}\alpha$.   
\item If $\tau_{first_{\alpha}} < \tau_{last_{\alpha}}$, and $\tau_i \in [\tau_{first_{\alpha}}+l, \tau_{last_{\alpha}}+u)$. 
 By the condition $\neg(\tau_i \in [\tau_j+u, \tau_k+l))$ for consecutive points $j,k$ where $\alpha$ holds, we have 
 for all consecutive points $j < k$ where $\alpha$ holds, $\tau_i \notin [\tau_j+u, \tau_k+l)$. Combining this 
 with  $\tau_i \in [\tau_{first_{\alpha}}+l, \tau_{last_{\alpha}}+u)$, 
  we have 
 $\tau_i \in [\tau_k+l, \tau_{last_{\alpha}}+u)$ or $\tau_i \in [\tau_{first_{\alpha}}+l, \tau_j+u)$ 
  for $k \leq last_{\alpha}$ and $j \geq first_{\alpha}$. 
 
   If $j=first_{\alpha}$, and if $\tau_i \in [\tau_{first_{\alpha}}+l, \tau_{first_{\alpha}}+u)$, and as seen in the first case, $\rho, i \models \past_{[l,u)}\alpha$. Similar is the case when $k=last_{\alpha}$. 
  Assume now that $j > first_{\alpha}$ and $k < last_{\alpha}$.
 Considering $j'$ as the immediate point before $j$ where $\alpha$ holds, (there is certainly such a point $j'$, $j'$ could be $first_{\alpha}$) we have by assumption
 $\tau_i \notin [\tau_{j'}+u, \tau_j+l)$. This combined with 
  $\tau_i \in [\tau_{first_{\alpha}}+l, \tau_j+u)$ 
  gives $\tau_i \in [\tau_j+l, \tau_j+u)$. 
  Similarly, considering $k'$ as the immediate next point after $k$ where $\alpha$ holds (there is certainly one such point, $k'$ could be $last_{\alpha}$) we have 
  by assumption $\tau_i \notin [\tau_k+u, \tau_{k'}+l)$. This combined with 
  $\tau_i \in  [\tau_k+l, \tau_{last_{\alpha}}+u)$ gives $\tau_i \in [\tau_k+l, \tau_k+u)$. Hence, we have $\rho,i \models \past_{[l,u)}\alpha$. 
  \end{enumerate}
\end{proof}

\subsection{Proof of Lemma \ref{remove-pastinf}}
\label{proofs-4}
\begin{proof}
Let $\rho'$ be a $(\Sigma,W)$-oversampled behaviour.  
 Let $\alpha=(act \rightarrow (\neg a \wedge \neg b))$.
Consider the following formulae in $\mtl$:
 \begin{enumerate}
  \item $\psi_1: [\wB \alpha \vee \{\alpha \wU[(a \wedge act) \wedge \wB_{[0,l)}(act \rightarrow \neg b)] \}]$
\item $\psi_2: \wB[(a \wedge act) \rightarrow \Box_{[l,\infty)}(act \rightarrow b)]$.
 \end{enumerate}
 Let $\psi=\psi_1 \wedge \psi_2$.  
 We claim that $\rho' \models ONF_{\Sigma}(T)$ iff $\rho' \models \psi$. 
   Assume $\rho' \models ONF_{\Sigma}(T)$. 
Assume the contrary that $\rho' \models \neg \psi_1$. Then, either there is a point marked $act \wedge b$
 before the first occurrence of $a \wedge act$, or there 
is a point marked $act \wedge b$ in the $[0,l)$ future of the first $a \wedge act$. 
Both of these imply $\neg ONF_{\Sigma}(T)$ giving contradiction.

 Now assume that  $\rho' \models \neg \psi_2$. Then 
some point $act$ in the $[l, \infty)$ future of a certain $a\wedge act$ 
is marked $\neg b$, which again contradicts $ONF_{\Sigma}(T)$.
 Hence $\rho' \models \psi$.  The converse can be proved in a similar way.
Note that $\psi_1 \wedge \psi_2$ increases a constant number of modalities compared to $ONF_{\Sigma}(T)$.
  %found in Appendix \ref{pastinf-c}. 

%Conversely, assume $\rho' \models \psi_1 \wedge \psi_2$. Let $\rho' \models \neg  \hat{X}_{[l, \infty)}$. 
 %Then, there is a point $i$ such that $\rho', i \models act$ and 
 %$\rho', i \nvDash (b \Leftrightarrow \past_{[l, \infty)} (a \wedge act))$. 
 %Assume that $\rho', i \models b$, but $\rho', i \nvDash \past_{[l, \infty)} (a \wedge act)$. 
 %Then, all points $act$ in $[0,t_i-l]$ are marked $\neg a$. Then, 
 %by $\varphi_1$, 
 %\begin{enumerate}
 % \item all points $act$ in the $[0,l)$ future 
 %of the first $a \wedge act$ must be marked $\neg b$
 %\item $\neg b \wedge \neg a$ holds at all points $act$ 
 %till the first $a \wedge act$. 
 %\end{enumerate}
%Given the above two points, we cannot have a $b$ at $t_i$. 
%Thus, $\rho',i  \nvDash \past_{[l, \infty)} (a \wedge act)$, and $\rho',i \models act$ 
%gives $\rho', i \models \neg b$. 

%Assume now that  $\rho', i \models \neg b \wedge act$. We
%then show that $\rho', i \models \neg \past_{[l, \infty)} (a \wedge act)$. 
%Assume the contrary : $\rho', i \models \past_{[l, \infty)} (a \wedge act)$. 
%Then there is a point marked $a \wedge act$ in $[0,t_i-l]$. 
%Then, by $\psi_2$, all points $act$ in $[t_i, \infty)$ are marked $b$, contradicting our assumption. 
\end{proof}

\subsection{Proof of Correctness of Lemma \ref{past-b1}}
\label{os-proof}
\begin{proof}
We give a proof of correctness on the construction of $\rho''$ and the formula 
   $\mathsf{MARK}$,  showing that 
   $ONF_{\Sigma}(T)=\exists \downarrow X.  ONF_{\Sigma_i}(\mathsf{MARK})$. 
   We start with a $(\Sigma,W)$-oversampled behaviour $\rho'$ over $\Sigma \cup W$.  
   We induct on the $a$'s in $\rho'$, and show that 
  a point $p$ of $\rho'$ is marked $b$ iff $\rho', p \models \past_{[l,u)}a$. 
  \begin{itemize}       
   \item Given any point $q$ of $\rho'$ marked $a$,  $\mathsf{MARK}_b$ marks all points in $[\tau_q+l,\tau_q+u)$  with $b$.  
         \item Lets look at the first $a$ of $\rho'$. Recall that 
    the point where $a$ holds for the first time is called $first_a$. 
    The formula $\mathsf{MARK}_{first}$ ensures that 
    all points of $\rho'$ that are at a distance $[0,l)$ from $first_a$ are marked $\neg b$. Also, all points 
    in $[0,\tau_{first_a}]$ are also marked $\neg b$. Thus, $\mathsf{MARK}_{first}$ accounts for all points 
    in $[0, \tau_{first_a}+l)$, while $\mathsf{MARK}_b$ marks all points in $[\tau_{first_a}+l,\tau_{first_a}+u)$ 
    with $b$. 
 \item Consider a point $j$ in $dom(\rho')$ such that $j > first_a$, $a \in \sigma_j$ and assume that 
all the $a$'s in $[0,\tau_j]$ have been accounted for:  that is,
all points in $[0, \tau_j+u)$ of $\rho'$ have been marked with $b$ or $\neg b$ correctly. 
This is the inductive hypothesis. 
Now consider the next consecutive $a$ occurring after  $j$, call that point $k$.
If $\tau_k - \tau_j \leq u-l$, then $\tau_k+l \leq \tau_j+u$, and 
by $\mathsf{MARK}_b$, all points in $[\tau_k+l,\tau_k+u)$ will be marked $b$. Hence, 
we are done accounting for $[0, \tau_k+u)$.
Hence, assume $\tau_k - \tau_j > u-l$. In this case, $\tau_k+u > \tau_k+l > \tau_j+u$.    
$\mathsf{MARK}_b$ marks all points in $[\tau_k+l,\tau_k+u)$ with $b$; 
we need to reason that points in $[\tau_j+u, \tau_k+l)$ will be marked $\neg b$. 
\begin{itemize}
 \item 
 The formulae $\mathsf{MARK}_{j,k}$, $\mathsf{MARK}_{beg,end}$ 
 mark  points $j,k$ respectively with $b_s,b_e$, and points $\tau_j+u, \tau_k+l$ respectively        
  with $beg_b$ and $end_b$. Also,   $\mathsf{MARK}_{beg,end}$
  marks $( \tau_j+u, \tau_j+u+1)$ as well as 
  $(\tau_j+u-1, \tau_j+u)$  with 
  $\neg beg_b$. As discussed in Lemma \ref{past-b1}, 
  we must assert that all other remaining points $beg_b$ and $end_b$  do not occur.
 The formula  $\mathsf{MARK}_c$ first marks all integer points with $c$. 
   We then identify the points between $b_s$ and $b_e$ 
  by uniquely marking the closest integral point before $b_s$ with $c_{b_s}$ and 
  and the closest integral point before $b_e$ with $c_{b_e}$. 
  Recall that $b_s$ and $b_e$ were marked at $\tau_j$ and $\tau_k$; 
  thus, $c_{b_s}$ and $c_{b_e}$ get marked respectively at points $\lfloor \tau_j \rfloor$ 
  and $\lfloor \tau_k \rfloor$.
  We then assert that $beg_b$ can occur at a point $t$ 
  iff there is a $c_{b_s}$ in $(t-u-1, t-u]$.  
  Thus, given that $c_{b_s}$ is marked at 
  $\lfloor \tau_j \rfloor$, $beg_b$ is marked only 
  in $[\lfloor \tau_j \rfloor+u, \lfloor \tau_j \rfloor+u+1)$.
  However, by formula $\mathsf{MARK}_{beg,end}$,
  we disallow $beg_b$ in $(\tau_j+u, \tau_j+u+1)$
  and $(\tau_j+u-1, \tau_j+u)$.
   Thus, we obtain a unique marking for $beg_b$.
  In a similar way, we obtain a unique marking for $end_b$.
  Note that the oversampled behaviour $\rho''$
  now has these markings. The formula 
  $\mathsf{MARK}_{\neg b}$ now marks all points of $\rho'$ (or all points 
    marked $act$ in $\rho''$) between $beg_b$ and $end_b$ with $\neg b$. 
    This takes care of the interval we were interested in: the interval 
    $[\tau_j+u, \tau_k+l)$.                        
      \item Thus, we have now accounted for all points of $\rho'$ in $[0, \tau_k+u)$.   
   \end{itemize}      
   \item We are now left with the remaining part $[\tau_k+u, \tau_{|dom(\rho')|}]$. 
   If $k \neq last_a$, we can extend the reasoning above to the next consecutive position 
   after $k$, which is marked an $a$. In this way, we can account for 
   all points of $\rho'$ in $[0, \tau_{last_a}+u)$. We just need to reason for $[\tau_{last_a}+u, \tau_{|dom(\rho')|}]$. 
  Consider the point $last_a$. 
  %Again,  by $\mathsf{MARK}_b$, all points in $[\tau_{last_a}+l,\tau_{last_a}+u)$ will be marked $b$.
   The formula $\mathsf{MARK}_{last}$ marks all points of $\rho'$ in the interval 
   $[\tau_{last_a}+u, \tau_{|dom(\rho')|}]$ with $\neg b$. 
   \end{itemize}      
 The above argument shows that all points of $\rho'$ are marked $b$ or $\neg b$ correctly. 
 The $(\Sigma \cup W, X)$-oversampled behaviour $\rho''$                 
   reflects these markings. When we do an oversampled projection 
   of $\rho''$ with respect to $X=\{b_e,b_s,beg_b,end_b,c,c_{b_s},c_{b_e}\}$, we are left with $\rho'$, 
   where at all positions, we have the correct marking with respect to $b$ or $\neg b$. 
  Clearly,  a point $p$ of $\rho'$ is marked $b$ iff $\rho', p \models \past_{[l,u)}a$.
   Hence,  $\rho' \models ONF_{\Sigma}(T)$ iff
   $\rho''  \models  ONF_{\Sigma \cup W \cup X}(\mathsf{MARK})$.

  Conversely, if we start with a $(\Sigma \cup W,X)$-oversampled behaviour $\rho''$ satisfying 
  $ONF_{\Sigma \cup W \cup X}(\mathsf{MARK})$, then all points 
  $p$ of   $\rho''$ marked $act$ will be marked $b$ iff 
  $\past_{[l,u)}a$ holds good at $p$. Then $\rho'' \downarrow X$ will 
  give a word $\rho'$ over $\Sigma \cup W$ that satisfies 
  $ONF_{\Sigma}(T)$.                         
 \end{proof}
\subsection{Extending Lemma \ref{past-b1} to weakly monotonic timed words}
\label{weak}
Note that for weakly monotonic words, we need to specify the exact location of $beg_b$ and $end_b$ for a fixed  time-stamp.
Recall that we mark the time stamp $\tau_j+u$ with $beg_b$, for a pair $j,k$ of consecutive $a$'s at distance 
$>u-l$.  
\begin{itemize}
\item Since there are several occurrences of the same time stamp, we want $beg_b$ to the first 
symbol of the repeating time stamp $\tau_j+u$ while dealing with   
intervals $\langle l,u)$.
We then add an extra formula $Fweak_{beg_b}=\wB (\Box_{[0,0]}\neg beg_b)$ which says that 
$beg_b$ is not after any symbol $\alpha$ having the same time stamp as $beg_b$.
  
\item Likewise, while dealing with intervals 
$\langle l,u]$,  $beg_b$ should always be the last symbol at its timestamp $\tau_j+u$. 
$Lweak_{beg_b}=\wB(beg_b \rightarrow \Box_{[0,0]}\bot)$ which says that there are no symbols 
$\alpha$ after $beg_b$ sharing the same time stamp as $beg_b$.
\item In a similar way, the position of $end_b$ depends on the left parantheses of the interval. Recall that we mark $end_b$ 
at $\tau_k+l$.  
If the interval is of the form $[l, u \rangle$, then we want  
$end_b$ to be the first symbol with time stamp $\tau_k+l$.
Similarly, if the interval is of the form $(l, u \rangle$, 
then we want $end_b$ to be the last symbol at time stamp $\tau_k+l$.
This can be done similarly as above.
\end{itemize}

\subsection{Proof of Lemma \ref{remove-pastinf2}}
 \label{proofs-5}
 \begin{proof}
 
 The temporal definition $T$ is the conjunction of  
 $C_1 = \wB[b \leftarrow \past_{[l, \infty)}a]$ and $C_2 =  \wB[b \rightarrow \past_{[l, \infty)}a]$. 
Models $\rho$ satisfying $C_1$ are those where 
 $b$ holds at all points $i$ such that $a$ holds somewhere 
 from the beginning of $\rho$ till $\tau_i-l$, that is in the prefix
  $[0, \tau_i-l]$ of $\rho$. Clearly,   
either there is no point marked $a$ in the model, in which case  
$\wB\neg a$ holds, or, 
whenever there is point $i$ marked $a$, then $b$ holds at all points 
in $[\tau_i+l, \infty)$. 
Thus, $C_1$ is equivalent to $\psi_1=\wB(\neg a) \vee 
\wB[a\rightarrow \Box_{[l,\infty)}b]$.

Models satisfying $C_2$ are those in which points where $\neg \past_{[l, \infty)}a$ hold must be marked $\neg b$. 
Clearly, all points in $[0,l)$ must be marked $\neg b$. Also,  if $i$ is the point where $a$ holds for the first time, 
then all points in $[\tau_i, \tau_i+l)$ should be marked $\neg b$. Thus, the formula 
$\psi_2=\wB(\neg a\wedge  \neg b) \vee (\neg a\wedge  \neg b)  \wU (a \wedge \wB_{[0,l)} \neg b)$ is equivalent to $C_2$. 
We thus have a formula $\psi_1 \wedge \psi_2 \in \mtl$ equivalent to $T$.
 \end{proof}

 \subsection{Proof of Correctness for Lemma \ref{past-b2}}
\label{correct}
\begin{proof}
The proof of correctness proceeds in similar lines as Lemma \ref{past-b1}. 
 We give a proof of correctness on the construction of $\rho'$ and the formula 
   $\mathsf{MARK}$,  showing that 
   $ENF_{\Sigma}(T)=\exists X. \mathsf{MARK}$. 
   We start with a timed word $\rho$ over $\Sigma \cup W$.    
   We induct on the $a$'s in $\rho$, and show that 
  a point $p$ of $\rho$ is marked $b$ iff $\rho, p \models \past_{[l,u)}a$. 
 \begin{itemize}       
   \item Given any point $q$ of $\rho$ marked $a$,  $\mathsf{MARK}_b$ marks all points in $[\tau_q+l,\tau_q+u)$  with $b$.  
         \item Lets look at the first $a$ of $\rho$. Recall that 
    the point where $a$ holds for the first time is called $first_a$. 
    The formula $\mathsf{MARK}_{first}$ ensures that 
    all points of $\rho$ that are at a distance $[0,l)$ from $first_a$ are marked $\neg b$. Also, all points 
    in $[0,\tau_{first_a}]$ are also marked $\neg b$. Thus, $\mathsf{MARK}_{first}$ accounts for all points 
    in $[0, \tau_{first_a}+l)$, while $\mathsf{MARK}_b$ marks all points in $[\tau_{first_a}+l,\tau_{first_a}+u)$ 
    with $b$. 
 \item Consider a point $j$ in $dom(\rho)$ such that $j > first_a$, $a \in \sigma_j$ and assume that 
all the $a$'s in $[0,\tau_j]$ have been accounted for:  that is,
all points in $[0, \tau_j+u)$ of $\rho$ have been marked with $b$ or $\neg b$ correctly. 
This is the inductive hypothesis. 
Now consider the next consecutive $a$ occurring from $j$, call that point $k$.
If $\tau_k - \tau_j \leq u-l$, then $\tau_k+l \leq \tau_j+u$, and 
by $\mathsf{MARK}_b$, all points in $[\tau_k+l,\tau_k+u)$ will be marked $b$. Hence, 
we are done accounting for $[0, \tau_k+u)$.
Hence, assume $\tau_k - \tau_j \in (u-l,u]$. In this case, $\tau_k+u > \tau_k+l > \tau_j+u$.    
$\mathsf{MARK}_b$ marks all points in $[\tau_k+l,\tau_k+u)$ with $b$; 
we need to reason that points in $[\tau_j+u, \tau_k+l)$ will be marked $\neg b$. 
\begin{itemize}
\item We start marking points of $\rho$ with new propositions, obtaining 
a simple extension $\rho'$ of $\rho$.  
We start marking points where $a$ holds good in $\rho$ with 
propositions in $\{a_0,a_1\}$. 
\item Assume that point $j$ is marked $a_0$, while $k$ is marked $a_1$ by formula 
$\mathsf{MARK}_a$. Let $d=\lceil \tau_k-\tau_k \rceil +l-u$, the 
closest integer $\geq$ the duration of the interval $[\tau_j+u, \tau_k+l)$.
Formula $\mathsf{MARK}_{beg,end,d}$ marks $j$ with $beg_{db}$ and point $k$ with $end_{db}$. 
Identifying point $j$ as $beg_{db}$ and point $k$ with $end_{db}$, 
all points in $I^2_{j,k}=[\tau_{end_{db}}+l-d,\tau_{end_{db}}+l)$ are marked $x_{0b}$ 
and all points in $I^1_{j,k}=[\tau_{beg_{db}}+u,\tau_{beg_{db}}+u+d)$ are marked $y_{0b}$.
The points in $I^1_{j,k} \cap I^2_{j,k}$ are marked $\neg b$ by $\mathsf{MARK}_{\neg b,0}$.  
\item Since $[\tau_j+u, \tau_k+l) \subseteq I^1_{j,k} \cap I^2_{j,k}$, 
we have clearly marked all points in  
$[\tau_j+u, \tau_k+l)$ with $\neg b$. Also, points 
in $[\tau_j+u, \tau_k+l)$ are not handled by formula 
$\mathsf{MARK}_b$, since these points are not in the $[l,u)$-future 
of any point marked $a$. Thus, points handled by 
$\mathsf{MARK}_{\neg b,0}$ and 
$\mathsf{MARK}_b$ are disjoint. 
\item Recall the discussion in Lemma \ref{past-b2} regarding free occurrences 
of $x_{1b},y_{1b}$ : as noted earlier, 
if $\{x_{1b},y_{1b}\} \subseteq \sigma_p$ for any $p \in  I^1_{j,k} \cap I^2_{j,k}$, 
there is no problem, since these points are anyway marked $\neg b$;
if $\{x_{1b},y_{1b}\} \subseteq \sigma_p$, for $p \notin  I^1_{j,k} \cap I^2_{j,k}$, 
then either they lie in some $I^1_{h,m} \cap I^2_{h,m}$ corresponding 
points $h, m$ such that $\tau_m-\tau_h \in (u-l,u]$, 
or $p$ is a point handled by 
$\mathsf{MARK}_b$. In the former case, there is no problem, while in the latter case, 
we get an inconsistent simple extension $\rho'$ from $\rho$. Since we 
work only on consistent simple extensions,
we rule out simple extensions 
where of the latter form.  
\item Thus, to summarize, 
we have accounted for all points $[0, \tau_k+u)$, being marked by one of $b, \neg b$
in consistent simple extensions. 
\end{itemize}

\item We are now left with the remaining part $[\tau_k+u, \tau_{|dom(\rho)|}]$. 
   If $k \neq last_a$, we can extend the reasoning above to the next consecutive position 
   after $k$, which is marked an $a$. In this way, we can account for 
   all points of $\rho$ in $[0, \tau_{last_a}+u)$. We just need to reason for $[\tau_{last_a}+u, \tau_{|dom(\rho)|}]$. 
  Consider the point $last_a$. 
  %Again,  by $\mathsf{MARK}_b$, all points in $[\tau_{last_a}+l,\tau_{last_a}+u)$ will be marked $b$.
   The formula $\mathsf{MARK}_{last}$ marks all points of $\rho$ in the interval 
   $[\tau_{last_a}+u, \tau_{|dom(\rho)|}]$ with $\neg b$. 
   \end{itemize}

 The above argument shows that all points of $\rho$ are marked $b$ or $\neg b$ correctly. 
 The $(\Sigma \cup W, X)$-simple extension $\rho'$                 
   reflects these markings. When we do a simple projection 
   of $\rho'$ with respect to $X$, we are left with $\rho$,
   the timed word over $\Sigma \cup W$ satisfying $ENF_{\Sigma}(T)$.
    On this $\rho$,  at all positions, we have the correct marking with respect to $b$ or $\neg b$. 
  Clearly,  a point $p$ of $\rho$ is marked $b$ iff $\rho, p \models \past_{[l,u)}a$.
   Hence,  $\rho \models ENF_{\Sigma}(T)$ iff
   $\rho'  \models \mathsf{MARK}$.

  Conversely, if we start with a timed word $\rho'$ over $\Sigma \cup W \cup X$ satisfying 
  $\mathsf{MARK}$, then any point 
  $p$ of   $\rho'$ will be marked $b$ iff 
  $\past_{[l,u)}a$ holds good at $p$. Then $\rho' \setminus X$ will 
  give a word $\rho$ over $\Sigma \cup W$ that satisfies 
  $ENF_{\Sigma}(T)$ iff $\rho'$ is a $(\Sigma \cup W,X)$-simple extension.                         

\end{proof}

\subsection{Simple Versus Oversampling Projections: Formulae Size}
\label{compl}
Consider a formula $\varphi \in \mathsf{MTL}[\until_I, \past_{np}]$. 
First we discuss the case of eliminating $\past_{np}$ by simple projections. 
Assume that the number of past modalities in $\varphi$ is $n$, of which there are  
$n_b$  bounded past modalities and $n_u$ unbounded past modalities. 
i.e, $n=n_b+n_u$. 
\begin{enumerate}
\item The first step is flattening, resulting in $\varphi_{flat}$. This only increases the size 
of the formula linearly in $n$. Converting $\varphi_{flat}$ to $ENF$ again increases the size 
by a constant number; thus, $ENF_{\Sigma}(\varphi_{flat})$ has a size increase of $\mathcal{O}(n)$ with respect to $\varphi$.
\item Let us first look at the $n_u$ unbounded past modalities. 
By Lemma \ref{remove-pastinf2}, the elimination 
of each temporal definition involving an unbounded past modality 
results in adding 2 formulae $\in \mtl$, 
and hence, in 3 extra modalities.  Thus, after elimination of all 
the $n_u$ temporal definitions, we get a formula whose size is increased by $\mathcal{O}(n)$.
\item Now let us look at the elimination of the temporal definitions 
corresponding to the $n_b$ bounded past modalities. 
\item Lemma \ref{past-b2}
deals with this.  Look at formula 2(a) (in Case 2) introduced by 
 Lemma \ref{past-b2}. This results in $l-1$ new formulae, and hence results in $\mathcal{O}(l)$ extra modalities.
   Thus, the number of extra modalities introduced after elimination 
   of all the $n_b$ temporal definitions corresponding to bounded past modalities is $\leq n l_{max}$, where 
   $l_{max}$ is the maximal lower bound of all bounded past modalities in $\varphi$.  
Assuming constants are encoded in binary,  $\mathcal{O}(n_b l_{max})$ is  pseudo polynomial; hence, the formula obtained by simple projections, $\psi_1$ has in the worst case, an exponential increase in  size over $\varphi$. 
Just to illustrate,  $l_{max}=10^{10}$ will really blow up! 
\item Note that Lemma \ref{past-b2} can further be optimized by changing the formula
2(a), 2(b), 3(a) and 3(b) in Case 2. 
Recall that formula 2(a)
is $\wB(x_{t+1+l-u~b} \leftrightarrow(a \wedge (\neg a \until_{(t,t+1]} a)))$, with $t \in \{u-l, \dots, u-1\}$, 
2(b) is $\wB(y_{db}\leftrightarrow(a \wedge (\neg a \since x_{db})))$, 
while formula 3(a) is  $\bigwedge_{c\in\{0,1\}}\wB((x_{db} \wedge a_c)\rightarrow \Box_{[u,u+d)} x_c)$ and
formula 3(b) is  
 $\bigwedge_{c\in\{0,1\}}\wB((y_{db} \wedge a_c)\rightarrow \wB_{[l-d,l)} y_{1-c}))$, where  $d \in \{1, \dots, l\}$. 
 The  ``bounding'' interval 
   between two consecutive $a$'s was considered as a unit interval here : 
   we were considering the interval lengths to lie in $(u-l, u-l+1], (u-l+1, u-l+2]$ and so on till
   $(u-1,u]$. This resulted in $l-1$ formulae. Had we chosen intervals of size  2 instead of 1,  
   we would 
   have considered the intervals as $(u-l, u-l+2], (u-l+2, u-l+4], \dots, 
   (u-2,u]$, resulting in $\frac{l}{2}$ formulae. 
      In general, we could have chosen as ``period'' any $\mu$ that gives rise to $\frac{l}{\mu}$ formulae. Clearly,
      since Case 2 in Lemma \ref{past-b2} considers $\tau_k-\tau_j \in (u-l,u]$, 
      the maximum period we can consider is $u-l$ i.e, $1 \leq \mu \leq u-l$. When $\mu=u-l$, we get $\frac{l}{u-l}$ formulae. In this case, replacing 2(a),2(b),3(a),3(b), we get  
  \begin{itemize}
\item  2(a) by
 ${\wB(x_{t+1+l-u~b} \leftrightarrow (a \wedge (\neg a \until_{(t,t+u-l]} a)))}$ for 
 ${t\in \{u-l,2(u-l), \ldots ,\frac{l(u-l)}{u-l}\}}$, 
\item 2(b) by $\wB(y_{\kappa b}\leftrightarrow(a \wedge (\neg a \since x_{\kappa b})))$
  \item 3(a) by  
  $\bigwedge_{c\in\{0,1\}}\wB((x_{\kappa b} \wedge a_c)\rightarrow \Box_{[u,u+\kappa (u-l))} x_c)$  
  \item 3(b) by \\
   $\bigwedge_{c\in\{0,1\}}\wB((y_{\kappa b} \wedge a_c)\rightarrow \wB_{[l-\kappa (u-l),l)} y_{1-c}))$ 
   where $\kappa\in \{1,\ldots,\frac{l}{u-l}\}$. 
   \end{itemize}
In this case, we get an increase of $\mathcal{O}(\frac{nl}{u-l})$ over the size of $\varphi$, 
as opposed to an increase of  $\mathcal{O}(n l_{max})$. Asymptotically, this is not a big saving, so we can stick to $\mu=1$.
\end{enumerate}
Now we discuss the case of  oversampled projections. 
Lemma \ref{remove-pastinf} discussed the case of unbounded past modalities and Lemma \ref{past-b1} the case of bounded past modalities. In both cases, it can be seen that the resultant formulae 
had an increase of size by a constant number, while eliminating 
each temporal definition. Thus, the total increase of size in the resultant formula $\psi_2 \in \mtl$ is only $\mathcal{O}(n)$.

\subsection{Eliminating $\since$ from $\mathsf{MTL}[\until_I,\since]$}
\label{since-rem}
Given a formula $\varphi \in \mtlu$ over $\Sigma$, we first flatten the formula to obtain  
formula $\varphi_{flat}$ over $\Sigma \cup W$.  In this section, we elaborate \cite{formats11}, \cite{deepak08} 
on removing the temporal definitions 
of the form $[r \leftrightarrow (c \since f)]$ from $\varphi_{flat}$, using future operators.
We use the short form $\nex \varphi$ to denote $false \until \varphi$.

$[r \leftrightarrow (c \since f)]$ will be replaced by a conjunction $\nu_r$ of the following future formulae:
\begin{itemize}
 \item $\varphi_1: \wB(f \rightarrow \nex r)$
 \item $\varphi_2 : \neg r$
 \item $\varphi_3 : \wB[(r \wedge c) \rightarrow \nex r]$
 \item $\varphi_4 : \wB[r \wedge (\neg c \wedge \neg f) \rightarrow \nex \neg r]$
 \item $\varphi_5 : \wB[(\neg r \wedge   \neg f) \rightarrow \nex \neg r]$
 \end{itemize}

For example, consider the formula \\
$\varphi=(a \wedge (b \wedge (c \until_{(1,2)}[(d \since e) \wedge f])))$ built from $\Sigma=\{a,b,c,d,e,f\}$.

The flattened version $\varphi_{flat}=(a \wedge b \wedge w_2) \wedge T_1 \wedge T_2$, where 
$T_1=\wB[(d \since e) \leftrightarrow w_1]$ and $T_2= \wB[w_2 \leftrightarrow c \until_{(1,2)}[w_1 \wedge f]]$.
$\varphi_{flat}$ is built from $\Sigma \cup W$, where $W=\{w_1,w_2\}$.  

Replace $T_1$  with $\nu_{w_1}$ to obtain 
the formula \\
$\psi=
(a \wedge b \wedge w_2) \wedge \nu_{w_1} \wedge  T_2 \in \mtl$. $\psi$ is also built from  $\Sigma \cup W$  and is  
equivalent to $\varphi_{flat}$. 
It can be seen that $\varphi=\exists W. \varphi_{flat}=\exists W.\psi$.

%trying to add example

\subsection{Proof of Lemma \ref{game:proof}}
\label{games}
We prove that the $\mtluns, \mtlsns$ are strictly less expressive than $\mtlfull$ using EF Games. 
We omit the game strategies here and give the candidate formula and pair of words.

\noindent (i) $\mtlfutpw \nsubseteq \mtluns$\\
We consider a formula in $MTL^{pw}[\fut_I]$,  $\varphi= \fut_{(0,1)}\{a \wedge \neg\fut_{[1,1]}(a\vee b)\}$. For an $n$-round game, 
consider the words $w_1=W_aW_b$ and $w_2=W_aW'_b$ with 
\begin{itemize}
 \item $W_a=(a, \delta)(a, 2 \delta) \dots (a,i\delta-\kappa)\underline{(a,i \delta)} \dots (a, n \delta)$
\item $W_b=(b, 1+\delta)(b, 1+2\delta) \dots (b,1+i\delta-\kappa)\underline{(b, 1+i\delta)} \dots (b, 1+n\delta)$
\item $W'_b=(b, 1+\delta)(b, 1+2\delta) \dots (b, 1+(i-1)\delta)(b,1+i\delta-\kappa)(b, 1+i\delta)(b, 1+(i+1)\delta) \dots (b, 1+n\delta)$ 
 \end{itemize}

$w_1 \models \varphi$, but $w_2 \nvDash \varphi$. 
The underlined $b$ in $W_b$ shows that there is a $b$ at distance 1 from $a$; however,
this is not the case with $W'_b$. 
The key observation for duplicator's win in an $\until_{NS}, \since_I$ game is that (a) any non-singular future move of spoiler can be mimicked by the duplicator from $W_aW_b$ or $W_a W'_b$ (b)
for any singular past move made by spoiler on $W_aW_b$, duplicator 
has  a reply from $W_aW'_b$. The same holds for any singular past move of spoiler made from 
$W_aW'_b$.\\

\noindent (ii) $\mtlfutp \nsubseteq \mtlsns$\\
 We consider a formula in $MTL^{pw}[\fut_I]$, $\phi' =  \fut \{b\wedge \neg \past_{[1,1]}(a \vee b)\}$. 
   We show that there is no way to express this formula in  $\mtlsns$. This is symmetrical to (i). For an $n$ round game, 
   consider the words 
   $w_1=W_aW_b$ and $w_2=W'_aW_b$ with 
\begin{itemize}
 \item $W_a=(a, \delta)(a, 2 \delta) \dots  (a, (i-1)\delta)(a, i\delta - \kappa)(a,i \delta)\dots (a, n \delta)$
 \item $W'_a=(a, \delta)(a, 2 \delta) \dots  (a, (i-1)\delta)(a,i \delta)\dots (a, n \delta)$
\item $W_b=(b, 1+\delta)(b, 1+2\delta) \dots (b, 1+(i-1)\delta)\underline{(b,1+i\delta-\kappa)}
(b, 1+i\delta) \dots (b, 1+n\delta)$
 \end{itemize}
   
 $w_1 \nvDash \varphi', w_2 \models \varphi'$.   
 The underlined $b$ in $W_b$ shows that there is an $a$ at past distance 1 
 in $W_a$, but not in $W'_a$. 
 The key observation for duplicator's win 
 in an $n$-round $\until_I, \since_{NS}$ game is that (a) any non-singular past move by spoiler from $W_a, W_b$ 
 or from $W'_a, W_b$ can be answered by duplicator, (b) for any singular future move 
   made by spoiler on $W_a, W_b$, duplicator 
has  a reply from $W'_a, W_b$. The same holds for any singular future move of spoiler made from 
$W'_a, W_b$.\\

\noindent (iii) $\mitlfp \nsubseteq \mathsf{MTL}[\until_I, \since]$.
We consider the $\mitlfp$ formula $\varphi''=\fut_{(1,2)}[a \wedge \neg \past_{(1,2)}a]$, and show that there is no  
way to express it using $\until_I, \since$. For an $n$ round game, consider the words
$w_1=W_1W_2$ and $w_2=W_1W'_2$ with
\begin{itemize}
 \item $W_1=(a, 0.5+\epsilon) \dots (a, 0.5+n\epsilon)(a,0.9+\epsilon)\dots(a,0.9+n\epsilon)$
 \item $W_2=\underline{(a, 1.5)}(a,1.6+\epsilon)(a, 1.6+2 \epsilon) \dots (a, 1.6+n \epsilon)$
 \item $W'_2=(a,1.6+\epsilon)(a, 1.6+2 \epsilon) \dots (a, 1.6+n \epsilon)$
\end{itemize}
for a very small $\epsilon >0$. Clearly, $w_1 \models \varphi'', w_2 \nvDash \varphi''$. 
The underlined $a$ in $W_2$ shows the $a$ in (1,2) which has no $a$ 
in $\past_{(1,2)}$. The key observation for duplicator's win 
 in an $n$-round $\until_I, \since$ game is that (a) when spoiler picks any position 
 in $W_1$, duplicator can play copy cat, (b) when spoiler  picks $(a, 1.5)$ in $W_2$ 
 as part of a future $(0,1)$ move from $W_1$, duplicator picks $0.9+n \epsilon$ in $W'_2$. All until, since 
 moves from the configuration $[(a.1.5),(a,0.9+n\epsilon)]$ are symmetric.

\end{document}